\documentclass[11pt,a4paper]{article}
\usepackage{amssymb,amsmath,amsthm}
\usepackage{color}
\usepackage[colorlinks,linkcolor=blue,citecolor=red]{hyperref}

\newcommand{\bbbr}{\mathbb R}

\newtheorem{theorem}{Theorem}

\newtheorem{prop}{Proposition}
\newtheorem{lemma}{Lemma}

\newtheorem{cor}{Corollary}
\newcommand{\comm}[1]{}

\renewcommand\a{\alpha}
\renewcommand\b{\beta}
\newcommand\ga{\gamma}

\newcommand\C{{\mathbb C}}

\renewcommand\d{\delta}
\newcommand\D{\mathcal{D}}

\newcommand\E{{\mathcal E}}

\newcommand\op[1]{\mathop{\rm #1}\nolimits}

\newcommand\p{\partial}

\newcommand\R{{\mathbb R}}

\begin{document}

\title{Involutive scroll structures on solutions of\\ 4D dispersionless integrable hierarchies}

\author{E.V. Ferapontov, B. Kruglikov}
     \date{}
     \maketitle
     \vspace{-5mm}
\begin{center}
{\small
Department of Mathematical Sciences \\ Loughborough University \\
Loughborough, Leicestershire LE11 3TU \\ United Kingdom\\
  \ \\  \ \\
Department of Mathematics and Statistics\\
Faculty of Science and Technology\\
UiT the Arctic University of Norway\\
Troms\o\ 90-37, Norway\\
[2ex]
e-mails: \\[1ex]  \texttt{E.V.Ferapontov@lboro.ac.uk}\\
\texttt{Boris.Kruglikov@uit.no} }

\bigskip

\end{center}

\begin{abstract}

A rational normal scroll structure on an $(n+1)$-dimensional  manifold $M$ is defined as a
field of rational normal scrolls of degree $n-1$ in the projectivised cotangent bundle $\mathbb{P} T^*M$.
We show that geometry of this kind naturally arises on solutions  of
various 4D dispersionless integrable hierarchies of  heavenly type equations. 
In this context, rational normal scrolls  coincide with the characteristic varieties (principal symbols) of the hierarchy.
Furthermore, such structures automatically satisfy an additional property of involutivity.

Our main result states that involutive scroll structures are themselves  governed by a dispersionless integrable hierarchy, namely, the hierarchy of conformal self-duality equations.

\bigskip


\end{abstract}

\newpage

\tableofcontents

\section{Introduction}

\subsection{Rational normal scrolls}
\label{RNS0}

In projective space $\mathbb{P}^n$ choose two complementary linear subspaces $\mathbb{P}^k$ and $\mathbb{P}^l$,
where $n=k+l+1$. 
Choose two rational normal curves of degrees $k$ and $l$ in these subspaces, and choose an isomorphism $\phi$ between them. 
Then the rational normal scroll $S_{k,l}$ is a ruled surface of degree $n-1$ consisting of all lines joining the corresponding 
points $x$ and $\phi(x)$, where $x$ varies over the first rational normal curve. Thus, $S_{k,l}$ is
a $\mathbb{P}^1$ bundle over $\mathbb{P}^1$.

We will assume $k,l>0$ in order for $S_{k,l}$ to be smooth (in the opposite case it is a cone).
If $k < l$ then the rational normal curve of degree $k$ is uniquely determined by the rational normal scroll 
and is called the directrix of the scroll \cite{GH}.

We recall that del Pezzo's theorem implies that scrolls $S_{k,l}$ are the only irreducible smooth ruled surfaces  of degree ${\rm deg} (S_{k,l})=k+l=n-1$ spanning $\mathbb{P}^n$. 

Rational normal scrolls $S_{k,l}$ are unambiguously characterized by their degrees and  automorphism groups.  
The latter are
 $$
\op{Aut}(S_{k,k})=PO_{2,2}\ \text{ for }k=l\quad\text{and}\quad
\op{Aut}(S_{k,l})=GL_2\ltimes\Bbbk^{l-k+1}\ \text{ for }k<l,
 $$
where $\Bbbk=\C$ as in classical algebraic geometry or $\Bbbk=\R$ as {
in the main part of our paper. We will review the necessary information about scrolls in Section \ref{sec:scrollsym}.

\subsection{Scroll structures on manifolds}

Given an $(n+1)$-dimensional manifold $M$, a (rational normal) scroll structure is a distribution of rational normal scrolls $S_{k,l}$ in the projectivised cotangent bundle  $\mathbb{P}^n=\mathbb{P}T^*M$ (for some $k, l$ such that $n=k+l+1$). Explicitly, a scroll structure can be parametrised as $\alpha\,\omega(\lambda)+\beta\,\phi (\lambda)$ where $\alpha$ and $\beta$ are some functions and
$\omega(\lambda),\ \phi(\lambda)$ are one-forms on $M$  polynomial in $\lambda$ of degree $k$ and $l$, respectively:
\begin{equation}\label{scroll}
\omega (\lambda)=\omega_0+\lambda \omega_1+ \dots + \lambda^{k}\omega_{k}, \qquad
\phi (\lambda)=\phi_0+\lambda \phi_1+ \dots + \lambda^{l}\phi_{l};
\end{equation}
here $\{\omega_i, \, \phi_j\}$ form  a basis of 1-forms  (a coframe) on $M$. The parameter $\lambda$ and the 1-forms $\omega(\lambda), \ \phi(\lambda)$ are defined up to transformations  $\lambda \to  \frac{a\lambda+b}{c\lambda +d}, \ \omega(\lambda) \to r(c\lambda+d)^{k}\omega(\lambda), \ \phi(\lambda) \to s(c\lambda+d)^{l}\phi(\lambda)$, where $a, b, c, d, r, s$  are functions on $M$
such that $ad-bc, r, s$ are nonzero. Without any loss of generality one can assume $ad-bc=1$. Furthermore, if $k<l$, we have an additional transformation freedom
$\phi(\lambda)\to \phi(\lambda)+p_{l-k}(\lambda)\omega(\lambda)$ where $p_{l-k}(\lambda)$ is a polynomial in $\lambda$ of degree $l-k$ whose coefficients are arbitrary functions on $M$.

Given a 4D dispersionless integrable hierarchy of heavenly type equations, it will be demonstrated  that the corresponding characteristic variety (zero locus of the principal symbol) determines canonically a scroll structure on every solution. It turns out that the scroll structures  arising in this way are not arbitrary and must satisfy an important property of {\it involutivity}.

\subsection{Involutive scroll structures (ISS)}
\label{sec:invol}

For every $x\in M$, the equations $\omega(\lambda)=\phi(\lambda)=0$ define a one-parameter family of 
linear subspaces of codimension two in $T_xM$ parametrised by $\lambda$; we will call them $\alpha$-subspaces.  
These subspaces are dual to the lines generating the scroll. A codimension two submanifold of $M$ is said
to be an {\it $\alpha$-manifold} if all its tangent spaces are $\alpha$-subspaces. The following definition is motivated by the theory of half-flat conformal structures \cite{Penrose}, involutive  $GL(2, \bbbr)$ structures \cite{Bryant2, Krynski,  FK1} and  half-flat Cayley structures \cite{KMa}.

\medskip

\noindent {\bf Definition.}
A scroll structure is said to be {\em involutive\/} 
if every $\alpha$-subspace is tangential to some $\alpha$-manifold.
We use abbreviation ISS for such structures.

\medskip

One can show that $\alpha$-manifolds of an involutive scroll structure depend on one arbitrary function of two
variables (see \cite{FK} for a similar statement in the context of involutive $GL(2)$-structures). The existence of  $\alpha$-manifolds suggests that involutive
scroll structures are amenable to twistor-theoretic methods.

\medskip

In the most general terms, the construction to be discussed in this paper can be summarised as follows. Let  $u$ be a (vector)-function on $(n+1)$-dimensional manifold $M$ (with local coordinates  $x^1,\dots, x^{n+1}$) that satisfies a hierarchy of compatible PDEs. In examples  below, the PDEs for $u$ will constitute an involutive system $\Sigma$ of second-order equations whose characteristic variety determines a scroll structure on every solution of $\Sigma$. The (dispersionless) integrability of $\Sigma$ will be understood as the existence 
of a  distribution $\D$ (Lax distribution) of codimension three on a $\mathbb{P}^1$-bundle $\hat M$ over $M$ (with local coordinates  $x^1,\dots, x^{n+1}, \lambda)$, such that $\D$ is integrable modulo $\Sigma$.
The scroll structures arising in this way are automatically involutive: the corresponding  $\alpha$-manifolds arise as projections to $M$ of the integral manifolds of $\D$ from $\hat M$.

Conventionally, the integrability of $\Sigma$ is expressed via $\lambda$-dependent vector fields by representing $\D$ 
as a span on $n-1$ vector  fields $V_i$ of the form 
$$
V_i=V_i^0+g_i^0(u, \lambda)\partial_{\lambda}, \qquad V_i^0=g_i^k(u, \lambda) \partial_k;
$$
here $\partial_k=\partial_{x^k}$ and the coefficients $g_i^j$ of $V_i$ depend on $u$, its derivatives and $\lambda$. 
The Frobenius  integrability condition, $[V_i, V_j]\in\D$, should be satisfied modulo $\Sigma$ (identically in 
$\lambda$). In other words, vector fields $V_i$ form an integrable distribution on every solution of $\Sigma$.

\subsection{Examples of ISS via the general heavenly hierarchy}
\label{Sghe}

The general heavenly hierarchy  is a collection of PDEs
 \begin{equation}
\begin{array}{c}
(a_k-a_j)(a_l-a_i)u_{jk}u_{il}+(a_i-a_k)(a_l-a_j)u_{ik}u_{jl}+(a_j-a_i)(a_l-a_k)u_{ij}u_{kl}=0,
\end{array}\label{gh}
 \end{equation}
one equation for every quadruple of distinct indices $i,j,k,l \in \{1, \dots, n+1\}$.  Here $u$ is a scalar function 
on $(n+1)$-dimensional manifold $M$ with local coordinates $x^i$, second-order partial derivatives of $u$ are 
denoted $u_{ij}=u_{x^ix^j}$, and $a_i$ are arbitrary constants. This hierarchy was constructed in \cite{Bogdanov}.

System (\ref{gh}) is equivalent to the requirement that the rank of the $(n+1)\times (n+1)$ skew-symmetric matrix  $\{(a_i-a_j)u_{ij}\}$ is equal to two. For any fixed values of indices $i,j,k,l$, equation (\ref{gh}) is known as the general heavenly equation; it was first derived in \cite{Schief} and later reappeared in \cite{DF} as the general case in the classification of 4D integrable symplectic Monge-Amp\'ere equations.
The characteristic variety of  system (\ref{gh}) is the intersection of  quadrics,
 \begin{equation}
 \begin{array}{rr}
(a_k-a_j)(a_l-a_i)(u_{jk}p_ip_l+u_{il}p_jp_k)
+(a_i-a_k)(a_l-a_j)(u_{ik}p_jp_l+u_{jl}p_ip_k)&\\[3pt]
{}+(a_j-a_i)(a_l-a_k)(u_{ij}p_kp_l+u_{kl}p_ip_j)&\!\!\!=0.
\label{ghchar}
 \end{array}
 \end{equation}
Equations (\ref{ghchar}) are equivalent to the requirement that the rank of the $(n+1)\times (n+1)$ skew-symmetric matrix  
$\{(a_i-a_j)(u_{ij}+p_ip_j)\}$ is equal to two.
They specify a smooth rational ruled surface of degree $n-1$ in 
$\mathbb{P}^{n}=\mathbb{P}T^*M$ which depends on the solution $u$, equivalently, a rational normal scroll. 
For $n=3$ we obtain a scroll $S_{1,1}$ which is nothing but a quadric in $\mathbb{P}^3$.  
For $n=4$ and $n=5$ we have the scrolls $S_{1,2}$ and $S_{2,2}$, respectively. In general, for $n=2k$ and $n=2k+1$, equations (\ref{ghchar}) define the scrolls $S_{k-1,k}$ and $S_{k,k}$, respectively.

 \begin{theorem}\label{ISSGPBE}
System \eqref{gh} is involutive and has involutive scroll structure on every solution.
Truncation of the hierarchy to dimension $n+1$ yields a family of ISS depending on $n-1$ functions of 3 variables.
 \end{theorem}

Leaving details of the proof to Section \ref{scrGHE},
let us explain how to obtain the ISS. Scrolls arise as characteristic varieties of the system 
depending on the background solution $u$ of \eqref{gh}. Introducing the 2-form
 $$
\Omega=\sum_{i,j}\left(\frac{1}{\lambda-a_i}-\frac{1}{\lambda-a_j}\right)u_{ij}\, dx^i\wedge dx^j,
 $$
one can represent equations (\ref{gh}) of the hierarchy in compact form as $d\Omega=0$, 
$\Omega\wedge\Omega=0$  \cite{Bogdanov}. The general heavenly hierarchy is also equivalent to involutivity of the 
$(n-1)$-dimensional distribution $\D$  spanned by $\lambda$-dependent vector fields
 $$
U_{ijk}=\left(\frac{1}{\lambda-a_i}-\frac{1}{\lambda-a_j}\right)u_{ij}\, \partial_{k}+\left(\frac{1}{\lambda-a_j}-\frac{1}{\lambda-a_k}\right)u_{jk}\, \partial_{i}+\left(\frac{1}{\lambda-a_k}-\frac{1}{\lambda-a_i}\right)u_{ik}\, \partial_{j},
 $$
which thus constitute a dispersionless Lax representation of the general heavenly hierarchy 
(note that in this particular case, the vector fields spanning $\D$ do not contain $\partial_{\lambda}$ terms). 
For every fixed $\lambda$, the integral manifolds of the above distribution are $\alpha$-manifolds. 
Thus, the corresponding scroll structure is involutive.

Representing $\Omega$ in the form $\Omega=r\wedge s$ with $r=\sum_i\frac{r_i}{\lambda-a_i}\, dx^i$,  
$s=\sum_i\frac{s_i}{\lambda-a_i}\, dx^i$, where $(a_i-a_j)u_{ij}=r_is_j-r_js_i$,
one can parametrise the characteristic variety  (\ref{ghchar}) in the form $p_i dx^i=\alpha\,r+\beta\,s$, explicitly, 
 \begin{equation}\label{rulingab}
p_i=\alpha\, \frac{ r_i}{\lambda-a_i}+\beta\, \frac{ s_i}{\lambda-a_i}.
 \end{equation}
Note that both 1-forms $r$ and $s$ annihilate the vector fields $U_{ijk}$ from the dispersionless Lax representation. Thus, distribution $\D$ is dual to the rulings of the scroll.

\subsection{Main results on general ISS in dimensions five and six}

The key observation of this paper is that involutive scroll structures arise, as characteristic varieties, on solutions  of dispersionless integrable hierarchies in 4D. 

In Section \ref{sec:ex} we discuss further examples of this kind,
in particular, ISS associated with the hierarchy of conformal self-duality equations of \cite{DFK, B2016, B2025}
in dimensions four, five and six. We will consider the general dimension in Section \ref{SIH}
for the second heavenly hierarchy and also revisit the general heavenly equation.

In Sections \ref{sec:alpha-par} and \ref{sec:alpha-par22} we obtain an explicit parametrisation of generic involutive $S_{1,2}$ and $S_{2,2}$ scroll structures via the first few flows of the hierarchy of conformal self-duality equations.
Our main results on these ISS are summarised as follows.

 \begin{theorem}\label{t12} 
A generic involutive $S_{1,2}$ scroll structure on a five-dimensional manifold can be  parametrised as
$\alpha\, \omega(\lambda)+\beta\, \phi(\lambda)$ where $\alpha, \beta$ are arbitrary functions and the one-forms 
$\omega(\lambda)$, $\phi(\lambda)$ are given by
 \begin{equation*}
 \begin{array}{c}
\omega(\lambda)=u_1dx^1+u_2dx^2+\lambda(dx^3-u_4dx^5),\\[5pt]
\phi(\lambda)=v_1dx^1+v_2dx^2+\lambda(dx^4-(v_4+w)dx^5)+\lambda^2dx^5.
 \end{array}
 \end{equation*}
Here $u, v$ and $w$ are functions of $x^1, \dots,  x^5$ that satisfy equations (\ref{alttriple}) of the hierarchy of conformal self-duality equations  coming from the requirement of involutivity of the distribution $\D=\langle Y_1, Y_2, Y_3\rangle$, 
where the vector fields $Y_i$ are given by 
 \begin{equation*}
 \begin{array}{c}
Y_1=\lambda \partial_1-u_{1}\partial_3-v_{1}\partial_4+\lambda w_1\partial_{\lambda},\
Y_2=\lambda \partial_2-u_{2}\partial_3-v_{2}\partial_4+\lambda w_2\partial_{\lambda},\\[5pt]
Y_3=\partial_5+u_{4}\partial_3+(v_{4}+w-\lambda)\partial_4-\lambda w_4 \partial_{\lambda}.
 \end{array}
 \end{equation*}
Generic involutive $S_{1,2}$ scroll structures locally depend on 9 arbitrary functions of 3 variables.

Conversely,  ISS is recovered from differential equations (\ref{alttriple}) as their characteristic variety.
This system of differential equations is integrable with dispersionless Lax representation $\langle Y_i\rangle$.
 \end{theorem}

 \begin{theorem}\label{t22} 
A generic involutive $S_{2,2}$ scroll structure on a six-dimensional manifold can be  parametrised as
$\alpha\, \omega(\lambda)+\beta\, \phi(\lambda)$ where $\alpha, \beta$ are arbitrary functions and the one-forms 
$\omega(\lambda)$, $\phi(\lambda)$ are given by
 \begin{equation*}
 \begin{array}{c}
\omega(\lambda)=u_1dx^1+u_2dx^2+\lambda(dx^3-u_4dx^5-(u_3+w)dx^6)+\lambda^2dx^6,\\[5pt]
\phi(\lambda)=v_1dx^1+v_2 dx^2+\lambda(dx^4-(v_4+w)dx^5-v_3dx^6)+\lambda^2dx^5.
 \end{array}
 \end{equation*}
Here $u, v$ and $w$ are functions of $x^1, \dots,  x^6$ that satisfy equations (\ref{alttriple})+(\ref{triple2}) of the hierarchy of conformal self-duality equations  coming from the requirement of involutivity of the distribution 
$\D=\langle Y_1, Y_2, Y_3, Y_4\rangle$ where the vector fields $Y_i$ are given by 
 \begin{equation*}
 \begin{array}{c}
Y_1=\lambda \partial_1-u_{1}\partial_3-v_{1}\partial_4+\lambda w_1\partial_{\lambda},\quad
Y_3=\partial_5+u_{4}\partial_3+(v_{4}+w-\lambda)\partial_4-\lambda w_4 \partial_{\lambda},\\[5pt]
Y_2=\lambda \partial_2-u_{2}\partial_3-v_{2}\partial_4+\lambda w_2\partial_{\lambda},\quad
Y_4=\partial_6+(u_3+w-\lambda)\partial_3+v_{3}\partial_4-\lambda w_3 \partial_{\lambda}.
 \end{array}
 \end{equation*}
Generic involutive $S_{2,2}$ scroll structures locally depend on 12 arbitrary functions of 3 variables.

Conversely,  ISS is recovered from equations (\ref{alttriple})+(\ref{triple2}) as their characteristic variety.
This system of differential equations is integrable with dispersionless Lax representation $\langle Y_i\rangle$.
 \end{theorem}

Finally, in Section \ref{sec:s}  we discuss symmetry aspects of involutive scroll structures including projective automorphisms of scrolls $S_{k,l}$ (Section \ref{sec:scrollsym}) and point symmetry algebras of the PDE systems governing involutive $S_{1,2}$ and $S_{2,2}$  scroll structures (Section \ref{sec:sym}).

Verifications of our computations, performed in symbolic packages of Maple, 
are available from the arXiv supplement to this paper.

\section{Examples of ISS in low dimensions}
\label{sec:ex}

In this section we provide further examples of involutive scroll structures associated with principal symbols of 
heavenly type hierarchies. We begin by revisiting the self-duality in four dimensions (note that a quadric 
in $\mathbb{P}^3$ is the simplest 
non-trivial scroll $S_{1,1}$, and its involutivity 
is equivalent to half-flatness). 
Then we concentrate on dimensions five and six, which provide entirely novel results. We will return to
the general dimension in Section \ref{SIH}.

\subsection{Two constructions in dimension four}
\label{S4D}

Let us begin with two constructions of half-flat (self-dual or anti-self-dual) metrics in 4D that will be 
needed in what follows. The first one is a recollection of a known result.

\medskip

{\noindent {\bf Construction 1.} \cite{DFK} Any half-flat metric in 4D can be represented in the form 
 \begin{equation}\label{prevSDM}
g=dx^2dx^3-dx^1dx^4-u_1(dx^3)^2+(v_1-u_2)dx^3dx^4+v_2(dx^4)^2,
 \end{equation}
where $x^1, \dots, x^4$ are local coordinates and the functions $u, v$, along with an auxiliary function $w$, 
satisfy a nonlinear system of second-order PDEs
 $$
Q_1(u)=w_1, \quad Q_1(v)=-w_2, \quad Q_1(w)=0,
 $$
where $Q_1$ is the second-order differential operator defined as
 $$
Q_1=\p_2\p_3-\p_1\p_4+u_1\p_2^2+(v_1- u_2)\p_1\p_2-v_2\p_1^2.
 $$
Note that $Q_1$, considered as a symmetric bi-vector, is the  inverse of the metric $g$ given by \eqref{prevSDM}.

The above system possesses a Lax pair in parameter-dependent vector fields, 
 \begin{equation}\label{C1}
 \begin{array}{c}
X_1=\partial_3+v_{1}\partial_1+u_{1}\partial_2-\lambda\partial_1+w_1\partial_{\lambda},\\[5pt]
X_2=\partial_4+v_{2}\partial_1+u_{2}\partial_2-\lambda\partial_2+w_2\partial_{\lambda}.
 \end{array}
 \end{equation}
This Lax pair has first appeared in \cite{BDM}. 

\medskip

{\noindent {\bf Construction 2.} This is apparently new: any half-flat metric in 4D can be represented in the form 
 \begin{equation}\label{newSDM}
g=(u_1dx^1+u_2dx^2)dx^4-(v_1dx^1+v_2dx^2)dx^3,
 \end{equation}
where $x^1, \dots, x^4$ are local coordinates and the functions $u, v$, along with an auxiliary function $w$  
(all different from that of Construction 1), satisfy a nonlinear system of second-order PDEs
 $$
R_1(u)=u_1w_2-u_2w_1, \quad R_1(v)=v_1w_2-v_2w_1, \quad R_1(w)=0,
 $$
where $R_1$ is the second-order differential operator defined as
 $$
R_1=u_1\p_2\p_3+v_1\p_2\p_4-u_2\p_1\p_3-v_2\p_1\p_4.
 $$
Note that $R_1$, considered as a symmetric bi-vector, is the inverse of the metric $g$ given by \eqref{newSDM}.

The above system possesses a Lax pair in parameter-dependent vector fields, 
 \begin{equation}\label{C2}
 \begin{array}{c}
Y_1=\lambda \partial_1-u_{1}\partial_3-v_{1}\partial_4+\lambda w_1\partial_{\lambda},\\[5pt]
Y_2=\lambda \partial_2-u_{2}\partial_3-v_{2}\partial_4+\lambda w_2\partial_{\lambda}.
 \end{array}
 \end{equation}

Construction 2 can be justified as follows: similarly to the proof of Theorem 2 of \cite{DFK}, we rectify two 
decomposable self-dual two forms (squares of pure spinors), namely,  $\Sigma_1=dx^1\wedge dx^2$ and 
$\Sigma_2=dx^3\wedge dx^4$. This means that the foliations $dx^1=0=dx^2$ and $dx^3=0=dx^4$ consist of 
$\alpha$-surfaces, whence $g$ has form \eqref{newSDM} with 
some functions $u_1,u_2,v_1,v_2$ of all four arguments (however not yet partial derivatives). 

The bundle of $\alpha$-planes lifted to the correspondence space has  form (\ref{C2}), where $w_1,w_2$ are again 
some functions of all four arguments. By the Penrose theorem, anti-self-duality of $g$ is equivalent to the involutivity 
of the distribution $\langle Y_1,Y_2\rangle$. The Frobenius condition, split in powers of $\lambda$, constitutes 
six equations, which tell that $u_1,u_2$ are derivatives of a function $u$ by $x^1$ and $x^2$, respectively, 
and similar for $v_1,v_2$ and $w_1,w_2$. Thus, we arrive at  the required form for the metric $g$, 
the Lax pair and the three second-order integrability conditions.

\medskip

Constructions 1, 2 can be viewed as deformations of the Plebansky second and first heavenly equations  \cite{Plebanski}, respectively. We will see that hierarchies of PDE systems  from  Constructions 1 and 2  (we will refer to them as  two equivalent versions of conformal self-duality equations), provide parametrisations of general involutive $S_{1,2}$ and $S_{2,2}$ scroll structures, see Sections \ref{sec:s12}, \ref{sec:s22} and \ref{sec:alpha-par}, \ref{sec:alpha-par22}.

\medskip

\subsection{Modified heavenly hierarchy and $S_{1,2}$ scroll structures}
\label{sec:h} 

The first three equations in the (modified) heavenly hierarchy  are as follows:
  \begin{equation} \label{heav}
u_{15}-u_{13}u_{44}+u_{14}u_{34}=0, \quad 
u_{14}u_{23}-u_{13}u_{24}=1, \quad 
u_{25}-u_{23}u_{44}+u_{24}u_{34}=0;
  \end{equation}
here $u$ is a function on 5-dimensional manifold $M$ with coordinates $x^1, \dots, x^5$. Note that the second equation is the well-known Plebanski's first heavenly equation \cite{Plebanski}, while the first and the third equations are the modified heavenly equations appearing in the classification of symplectic Monge-Amp\'ere equations in 4D \cite{DF}.
The characteristic variety is the intersection of quadrics,
  $$
\begin{array}{c}
p_1p_5-u_{13}p_4^2-u_{44}p_1p_3+u_{14}p_3p_4+u_{34}p_1p_4=0,\\[3pt]
u_{14}p_2p_3+u_{23}p_1p_4-u_{13}p_2p_4-u_{24}p_1p_3=0,\\[3pt]
p_2p_5-u_{23}p_4^2-u_{44}p_2p_3+u_{24}p_3p_4+u_{34}p_2p_4=0,
\end{array}
  $$
which specify a rational normal scroll $S_{1,2}$ parametrised as
$p_i dx^i=\alpha\, \omega(\lambda)+\beta\, \phi(\lambda)$,
where $\alpha, \beta$ are arbitrary functions and $\omega, \phi$ are 1-forms 
depending on auxiliary parameter $\lambda$:
  \begin{equation}\label{modh}
  \begin{array}{c}
\omega(\lambda)=-u_{14}dx^1-u_{24}dx^2+\lambda (dx^3 +u_{44} dx^5), \\[5pt]
\phi(\lambda)=u_{13}dx^1+u_{23}dx^2+\lambda(dx^4-u_{34}dx^5)+\lambda^2 dx^5.
  \end{array}
  \end{equation}
As $\omega(\lambda)$ is linear in $\lambda$, it plays the role of directrix of the scroll. The equations $\omega(\lambda)=\phi(\lambda)=0$ define a 3-dimensional distribution $\D$,
 $$
\D=\langle
\lambda \partial_1+u_{14}\partial_3-u_{13}\partial_4, \ \ 
\lambda \partial_2+u_{24}\partial_3-u_{23}\partial_4, \ \
\partial_5-u_{44}\partial_3+(u_{34}-\lambda)\partial_4 \rangle, 
 $$
which constitutes dispersionless Lax representation for  equations (\ref{heav}), that is, $\D$ is integrable modulo  (\ref{heav}) for every $\lambda$. 
Since system (\ref{heav}) is in involution and its characteristic variety (scroll $S_{1,2}$) has degree 3 and affine dimension 3,   its general solution  depends on 3 arbitrary functions of 3 variables.

\subsection{Hierarchy of conformal self-duality equations and $S_{1,2}$ scroll structures}
\label{sec:s12}

We will utilise two constructions of the hierarchy of conformal self-duality equations. 
\medskip

\subsubsection{Construction 1}
\label{s12c1}

The  first three Lax equations of this hierarchy  are given by the vector fields
  \begin{equation}\label{Laxh}
  \begin{array}{c}
X_1=\partial_3+v_{1}\partial_1+u_{1}\partial_2-\lambda\partial_1+w_1\partial_{\lambda},\\[3pt]
X_2=\partial_4+v_{2}\partial_1+u_{2}\partial_2-\lambda\partial_2+w_2\partial_{\lambda},\\[3pt]
X_3=\partial_5+v_{3}\partial_1+u_{3}\partial_2-w\partial_1-\lambda\partial_3+w_3\partial_{\lambda};
  \end{array}
  \end{equation}
here $u, v$ and $w$ depend on $x^1, \dots,  x^5$.  Note that the first two vector fields  coincide with (\ref{C1}). 
This hierarchy  was constructed in (\cite{B2016},  eqs. 15-16 and  top of page 6; see also \cite{B2025}, eqs. 38-39). 
Introducing the second-order operators $Q_1,Q_2,Q_3$ as
  \begin{equation}\label{Q}
  \begin{array}{c}
Q_1=\partial_2\partial_3-\partial_1\partial_4+u_1\partial_{2}^2+(v_1- u_2)\partial_{1}\partial_{2}- v_2\partial_{1}^2,\\[3pt]
Q_2=\partial_3^2-\partial_1\partial_5+u_1\partial_2\partial_3+v_1\partial_1\partial_3-u_3\partial_1\partial_2+(w-v_3)\partial_1^2, \\[3pt]
Q_3=\partial_3\partial_4-\partial_2\partial_5+u_2\partial_2\partial_3-u_3\partial_2^2+v_2\partial_1\partial_3+(w-v_3)\partial_1\partial_2,
  \end{array} 
  \end{equation}
one can represent the commutativity conditions $[X_1, X_2]=0$, $[X_1, X_3]=0$ and $[X_2, X_3]=0$ 
in compact form as
  \begin{equation}\label{triple}
  \begin{array}{c}
Q_1(u)=w_1, \quad Q_1(v)=-w_2, \quad Q_1(w)=0, \\[3pt]
Q_2(u)=-u_1w_1, \quad Q_2(v)=u_1w_2,\quad Q_2(w)=-w_1^2, \\[3pt]
Q_3(u)=-w_3-u_1w_2, \quad Q_3(v)=w_4+v_2w_1-(v_1-u_2)w_2, \quad Q_3(w)=-w_1w_2,
  \end{array}
  \end{equation}
respectively. Note that the first three of these equations coincide with the 4D conformal self-duality equations 
of \cite{DFK}. The characteristic variety of the full system (\ref{triple}) is the triple scroll $S_{1,2}$, where 
the scroll $S_{1,2}$  is defined by the equations
  \begin{equation}\label{p12}
  \begin{array}{c}
p_2p_3-p_1p_4+u_1p_2^2+(v_1-u_2)p_1p_2-v_2p_1^2=0,\\[3pt]
p_3^2-p_1p_5+u_1p_2p_3+v_1p_1p_3-u_3p_1p_2+(w-v_3)p_1^2=0,\\[3pt]
p_3p_4-p_2p_5+u_2p_2p_3-u_3p_2^2+v_2p_1p_3+(w-v_3)p_1p_2=0.
  \end{array}
  \end{equation}
It can be parametrised as
$p_i dx^i=\alpha\, \omega(\lambda)+\beta\, \phi(\lambda)$
where $\alpha, \beta$ are arbitrary functions and $\omega(\lambda), \phi(\lambda)$ are 1-forms depending on auxiliary parameter $\lambda$:
  \begin{equation}\label{omega12}
  \begin{array}{c}
\omega(\lambda)=dx^2-u_1dx^3-u_2dx^4-u_3dx^5+\lambda (dx^4-u_1dx^5),\\[5pt]
\phi(\lambda)=dx^1-v_1dx^3-v_2dx^4-(v_3-w)dx^5+\lambda (dx^3-v_1 dx^5)+\lambda^2 dx^5.
  \end{array}
  \end{equation}
As $\omega(\lambda)$ is linear in $\lambda$, it plays the role of directrix of the scroll. Since system (\ref{triple}) is in involution and its characteristic variety (triple scroll $S_{1,2}$) has degree   9 and affine dimension 3, its general solution  depends on 9 arbitrary functions of 3 variables.
The Lax distribution $\D=\langle X_1, X_2, X_3\rangle$ satisfies the equations $\omega(\lambda)=\phi(\lambda)=0$.

The hierarchy of 4D conformal self-duality equations possesses reductions to two known
hierarchies that we discuss in more detail below.

\medskip

\noindent{\bf  (a) Reduction to the hierarchy of Dunajski equations}. The hierarchy of conformal self-duality equations is compatible with the constraint $u_2+v_1=0$. Setting $u=-\theta_1, v=\theta_2$, equations (\ref{triple}) simplify to
  \begin{equation}\label{Dun}
  \begin{array}{c}
\theta_{14}-\theta_{23}+\theta_{11}\theta_{22}-\theta_{12}^2=w, \\[3pt]
 w_{14}-w_{23}+\theta_{11}w_{22}-2\theta_{12}w_{12}+\theta_{22}w_{11}=0,\\[6pt]
\theta_{15}-\theta_{33}+\theta_{11}\theta_{23}-\theta_{12}\theta_{13}=w\theta_{11}, \\[3pt]
 w_{15}-w_{33}+\theta_{11}w_{23}+\theta_{23}w_{11}-\theta_{12}w_{13}-\theta_{13}w_{12}
 =ww_{11}+w_1^2,\\[6pt]
\theta_{25}-\theta_{34}+\theta_{12}\theta_{23}-\theta_{22}\theta_{13}=w\theta_{12}-\rho, \\[3pt]
 w_{25}-w_{34}+\theta_{12}w_{23}+\theta_{23}w_{12}-\theta_{22}w_{13}-\theta_{13}w_{22}=ww_{12}+w_1w_2.
  \end{array}
  \end{equation}
Here $\rho$  is a non-local variable defined via the equations
  \begin{equation}\label{rho}
  \begin{array}{c}
\rho_1=w_3+\theta_{12}w_1-\theta_{11}w_2, \\[3pt]
\rho_2=w_4+\theta_{22}w_1-\theta_{12}w_2, \\[3pt]
\rho_3 = w_5 + (\theta_{23}-w)w_1 - \theta_{13}w_2,
  \end{array}
  \end{equation}
that are compatible modulo $(\ref{Dun})$. 
The first two equations (\ref{Dun}) were first introduced in \cite{DP} in the description of self-dual metrics with a parallel real spinor. The corresponding hierarchy was constructed in \cite{BDM}. The  characteristic variety of system (\ref{Dun})+(\ref{rho}) is a double scroll $S_{1,2}$ (modulo components of lower dimension), where $S_{1,2}$ is defined by equations (\ref{p12}) with the substitution $u=-\theta_1, v=\theta_2$.
Under the same substitution, the scroll $S_{1,2}$ is parametrised  via (\ref{omega12}). Since system (\ref{Dun})+(\ref{rho}) is in involution and its characteristic variety (double scroll $S_{1,2}$) has degree 6 and affine dimension 3,   its general solution  depends on 6 arbitrary functions of 3 variables.

\medskip

\noindent{\bf (b) Reduction to the hierarchy of the second heavenly equation.} There exists further reduction to the second heavenly hierarchy via $u_2+v_1=w=0$. Setting $u=-\theta_1, v=\theta_2$, one can rewrite the first three dispersionless Lax equations in the form
  $$
\begin{array}{c}
X_1=\partial_3+\theta_{12}\partial_1-\theta_{11}\partial_2-\lambda\partial_1,\\[3pt]
X_2=\partial_4+\theta_{22}\partial_1-\theta_{12}\partial_2-\lambda\partial_2,\\[3pt]
X_3=\partial_5+\theta_{23}\partial_1-\theta_{13}\partial_2-\lambda\partial_3;
\end{array}
  $$
their commutativity conditions result in the first three equations of the second heavenly hierarchy,
  \begin{equation}\label{2heav}
  \begin{array}{c}
\theta_{14}-\theta_{23}+\theta_{11}\theta_{22}-\theta_{12}^2=0, \\[3pt]
\theta_{15}-\theta_{33}+\theta_{11}\theta_{23}-\theta_{12}\theta_{13}=0, \\[3pt]
\theta_{25}-\theta_{34}+\theta_{12}\theta_{23}-\theta_{22}\theta_{13}=0.
  \end{array}
  \end{equation}
The  corresponding characteristic variety is a scroll $S_{1,2}$ parametrised as
$p_i dx^i=\alpha\, \omega(\lambda)+\beta\, \phi(\lambda)$ where $\alpha,\beta$ are arbitrary functions 
and $\omega, \phi$ are 1-forms depending on auxiliary parameter $\lambda$:
  $$
\begin{array}{c}
\omega(\lambda)=dx^2+\theta_{11}dx^3+\theta_{12}dx^4+\theta_{13}dx^5+\lambda(dx^4+\theta_{11}dx^5),
\\[5pt]
\phi(\lambda)=dx^1-\theta_{12}dx^3-\theta_{22} dx^4-\theta_{23}dx^5+\lambda(dx^3-\theta_{12}dx^5)+\lambda^2dx^5.
\end{array}
  $$
Since system (\ref{2heav}) is in involution and its characteristic variety (scroll $S_{1,2}$) has degree 3 and affine dimension 3,   its general solution  depends on 3 arbitrary functions of 3 variables.

\subsubsection{Construction 2}
\label{s12c2}

The first three Lax equations of an alternative form of the  hierarchy of conformal self-duality equations are given by the vector fields
  \begin{equation}\label{altLaxh}
  \begin{array}{c}
Y_1=\lambda \partial_1-u_{1}\partial_3-v_{1}\partial_4+\lambda w_1\partial_{\lambda},\\[3pt]
Y_2=\lambda \partial_2-u_{2}\partial_3-v_{2}\partial_4+\lambda w_2\partial_{\lambda},\\[3pt]
Y_3=\partial_5+u_{4}\partial_3+(v_{4}+w-\lambda)\partial_4-\lambda w_4 \partial_{\lambda};
  \end{array}
  \end{equation}
here $u, v$ and $w$ depend on $x^1, \dots,  x^5$. Note that the first two vector fields  coincide with (\ref{C2}). 
Introducing the second-order operators $R_1,R_2,R_3$ as
  \begin{equation}\label{altQ}
  \begin{array}{c}
R_1=u_1\partial_2\partial_3+v_1\partial_2\partial_4-u_2\partial_{1}\partial_{3}  - v_2\partial_{1}\partial_4,\\[3pt]
R_2=u_1\partial_3\partial_4+v_1\partial_4^2-\partial_1\partial_5-u_4\partial_1\partial_3-(v_4+w)\partial_1\partial_4, \\[3pt]
R_3=u_2\partial_3\partial_4+v_2\partial_4^2-\partial_2\partial_5-u_4\partial_2\partial_3-(v_4+w)\partial_2\partial_4,
  \end{array}
  \end{equation}
one can represent the involutivity conditions of the Lax distribution $\D=\langle Y_1, Y_2, Y_3\rangle$  
in compact form as
  \begin{equation}\label{alttriple}
  \begin{array}{c}
R_1(u)=u_1w_2-u_2w_1, \quad R_1(v)=v_1w_2-v_2w_1, \quad R_1(w)=0,\\[3pt]
R_2(u)=u_1w_4, \quad R_2(v)=-u_1w_3,\quad R_2(w)=w_1w_4, \\[3pt]
R_3(u)=u_2w_4, \quad R_3(v)=-u_2w_3, \quad R_3(w)=w_2w_4.
  \end{array}
  \end{equation}
The characteristic variety of system (\ref{alttriple}) is the triple scroll $S_{1,2}$, where the scroll $S_{1,2}$  
is defined by the equations
  \begin{equation}\label{altp12}
  \begin{array}{c}
u_1p_2p_3+v_1p_2p_4-u_2p_1p_3-v_2p_1p_4=0,\\[3pt]
u_1p_3p_4+v_1p_4^2-p_1p_5-u_4p_1p_3-(v_4+w)p_1p_4=0,\\[3pt]
u_2p_3p_4+v_2p_4^2-p_2p_5-u_4p_2p_3-(v_4+w)p_2p_4=0.
  \end{array}
  \end{equation}
It can be parametrised as
$p_i dx^i=\alpha\, \omega(\lambda)+\beta\, \phi(\lambda)$ where 
  \begin{equation}\label{altomega12}
  \begin{array}{c}
\omega(\lambda)=u_1dx^1+u_2dx^2+\lambda(dx^3-u_4dx^5),\\[5pt]
\phi(\lambda)=v_1dx^1+v_2dx^2+\lambda(dx^4-(v_4+w)dx^5)+\lambda^2dx^5.
  \end{array}
  \end{equation}
As $\omega(\lambda)$ is linear in $\lambda$, it plays the role of directrix of the scroll. Since system (\ref{alttriple}) is in involution and its characteristic variety (triple scroll $S_{1,2}$) has degree   9 and affine dimension 3, its general solution  depends on 9 arbitrary functions of 3 variables, the same as in Section \ref{s12c1}. 
The Lax distribution $\D$ satisfies the equations $\omega(\lambda)=\phi(\lambda)=0$.

It will be shown in Section \ref{sec:alpha-par} that system (\ref{alttriple}) gives a local parametrisation of {\it all} involutive $S_{1,2}$ scroll structures. A simple parameter count (9 arbitrary functions of 3 variables) suggests that this parametrisation must be equivalent to that obtained in Section \ref{s12c1}, however, direct transformation establishing equivalence of systems (\ref{triple}) and (\ref{alttriple}) is a nonlocal (B\"acklund-type) transformation.

\subsection{Hierarchy of conformal self-duality equations and $S_{2,2}$ scroll structures}
\label{sec:s22}

Again, we will utilise both constructions of the hierarchy of conformal self-duality equations. 

\subsubsection{Construction 1}
\label{s22c1}

The first four Lax equations of the hierarchy of 4D conformal self-duality equations  are given by the vector fields
  \begin{equation}\label{Laxh1}
  \begin{array}{c}
X_1=\partial_3+v_{1}\partial_1+u_{1}\partial_2-\lambda\partial_1+w_1\partial_{\lambda},\\[3pt]
X_2=\partial_4+v_{2}\partial_1+u_{2}\partial_2-\lambda\partial_2+w_2\partial_{\lambda},\\[3pt]
X_3=\partial_5+v_{3}\partial_1+u_{3}\partial_2-w\partial_1-\lambda\partial_3+w_3\partial_{\lambda},\\[3pt]
X_4=\partial_6+v_{4}\partial_1+u_{4}\partial_2-w\partial_2-\lambda\partial_4+w_4\partial_{\lambda},
  \end{array}
  \end{equation}
see \cite{B2016, B2025}. Introducing, in addition to the operators $Q_1, Q_2, Q_3$ given by (\ref{Q}), the  second-order  operators $Q_4,Q_5,Q_6$ as
  $$
\begin{array}{c}
Q_4=\partial_3\partial_4-\partial_1\partial_6+u_1\partial_{2}\partial_4+v_1\partial_{1}\partial_{4}+(w-u_4)\partial_1\partial_2  - v_4\partial_{1}^2,\\[3pt]
Q_5=\partial_4^2-\partial_2\partial_6+u_2\partial_2\partial_4+v_2\partial_1\partial_4+(w-u_4)\partial_2^2-v_4\partial_1\partial_2, \\[3pt]
Q_6=\partial_4\partial_5-\partial_3\partial_6+(v_3-w)\partial_1\partial_4+(w-u_4)\partial_2\partial_3+u_3\partial_2\partial_4-v_4\partial_1\partial_3,
\end{array}
  $$
one can represent the commutativity conditions $[X_1, X_4]=0$, $[X_2, X_4]=0$ and $[X_3, X_4]=0$ 
in compact form as
  \begin{equation}\label{triple1}
  \begin{array}{c}
Q_4(u)=w_3+(v_1-u_2)w_1+u_1w_2, \quad Q_4(v)=-w_4-v_2w_1, \quad Q_4(w)=-w_1w_2,\\[6pt]
Q_5(u)=v_2w_1, \quad Q_5(v)=-v_2w_2,\quad Q_5(w)=-w_2^2, \\[6pt]
Q_6(u)=w_5+(v_3-w)w_1+u_1w_4-u_2w_3+u_3w_2, \\[3pt]
Q_6(v)=-w_6+v_1w_4-v_2w_3+(w-u_4)w_2-v_4w_1, \quad
Q_6(w)=w_1w_4-w_2w_3,
  \end{array}
  \end{equation}
respectively. Equations (\ref{triple1}) must be considered together with equations (\ref{triple}) that come from the commutativity conditions of the first three vector fields, $X_1, X_2, X_3$. The characteristic variety of the full system (\ref{triple})+(\ref{triple1}) is the triple scroll $S_{2,2}$, where the scroll $S_{2,2}$  is defined by the equations
  \begin{equation}\label{p22}
  \begin{array}{c}
p_2p_3-p_1p_4+u_1p_2^2+(v_1-u_2)p_1p_2-v_2p_1^2=0, \\[3pt]
p_3^2-p_1p_5+u_1p_2p_3+v_1p_1p_3-u_3p_1p_2+(w-v_3)p_1^2=0, \\[3pt]
p_3p_4-p_2p_5+u_2p_2p_3-u_3p_2^2+v_2p_1p_3+(w-v_3)p_1p_2=0, \\[3pt]
p_3p_4-p_1p_6+u_1p_2p_4+v_1p_1p_4+(w-u_4)p_1p_2  - v_4p_1^2=0, \\[3pt]
p_4^2-p_2p_6+u_2p_2p_4+v_2p_1p_4+(w-u_4)p_2^2-v_4p_1p_2=0,  \\[3pt]
p_4p_5-p_3p_6+(v_3-w)p_1p_4+(w-u_4)p_2p_3+u_3p_2p_4-v_4p_1p_3=0.
  \end{array}
  \end{equation}
It can be parametrised as
$p_i dx^i=\alpha\, \omega(\lambda)+\beta\, \phi(\lambda)$
where $\alpha, \beta$ are arbitrary functions and $\omega(\lambda), \phi(\lambda)$ are 1-forms depending on auxiliary parameter $\lambda$:
  \begin{equation}\label{omega221}
  \begin{array}{rl}
\hskip2pt \omega(\lambda) &\!\!\!\!= dx^1-v_2dx^4+(w-v_3)dx^5-(v_1v_2+v_4)dx^6 
+\lambda(dx^3+v_1dx^5-v_2dx^6)+2\lambda^2dx^5,\hskip-12pt \\[6pt]
\hskip2pt \phi(\lambda) &\!\!\!\!= dx^2-u_1 dx^3+(v_1-u_2)dx^4-(u_1v_1+u_3)dx^5+(v_1^2-u_2v_1-u_4+w)dx^6
\hskip-12pt \\[3pt]
&\ \ {}+ \lambda(dx^4-u_1dx^5+(2v_1-u_2)dx^6)+2\lambda^2dx^6.
  \end{array}
  \end{equation}
Note that both $\omega(\lambda)$ and $\phi(\lambda)$ are quadratic in $\lambda$. Since system (\ref{triple})+(\ref{triple1}) is in involution and its characteristic variety (triple scroll $S_{2,2}$) has degree 12 and affine dimension 3, its general solution  depends on 12 arbitrary functions of 3 variables.
The Lax distribution $\D=\langle X_1, X_2, X_3, X_4\rangle$ satisfies the equations $\omega(\lambda)=\phi(\lambda)=0$.

The above equations possess analogous reductions to the Dunajski hierarchy ($u_2+v_1=0$) and the second heavenly hierarchy ($u_2+v_1=w=0$).


\subsubsection{Construction 2}
\label{s22c2}

The first four Lax equations of an alternative form of the  hierarchy of conformal self-duality equations are given by the vector fields
  \begin{equation}\label{altLaxh22}
  \begin{array}{c}
Y_1=\lambda \partial_1-u_{1}\partial_3-v_{1}\partial_4+\lambda w_1\partial_{\lambda}, \\[3pt]
Y_2=\lambda \partial_2-u_{2}\partial_3-v_{2}\partial_4+\lambda w_2\partial_{\lambda}, \\[3pt]
Y_3=\partial_5+u_{4}\partial_3+(v_{4}+w-\lambda)\partial_4-\lambda w_4 \partial_{\lambda}, \\[3pt]
Y_4=\partial_6+(u_3+w-\lambda)\partial_3+v_{3}\partial_4-\lambda w_3 \partial_{\lambda}.
  \end{array}
  \end{equation}
Introducing, in addition to the operators $R_1, R_2, R_3$ given by (\ref{altQ}), the  second-order  operators 
$R_4, R_5, R_6$ as
  $$
\begin{array}{c}
R_4=v_1\partial_3\partial_4+u_1\partial_3^2-\partial_1\partial_6-v_3\partial_1\partial_4-(u_3+w)\partial_1\partial_3, \\[3pt]
R_5=v_2\partial_3\partial_4+u_2\partial_3^2-\partial_2\partial_6-v_3\partial_2\partial_4-(u_3+w)\partial_2\partial_3, \\[3pt]
R_6=\partial_3\partial_5-\partial_4\partial_6+u_4\partial_3^2-v_3\partial_4^2+(v_4-u_3)\partial_3\partial_4,
\end{array}
  $$
one can represent the involutivity conditions $[Y_1, Y_4]\in\D$, $[Y_2, Y_4]\in\D$ and 
$[Y_3, Y_4]\in\D$ in compact form as
  \begin{equation}\label{triple2}
  \begin{array}{c}
R_4(u)=-v_1w_4, \quad R_4(v)=v_1w_3, \quad R_4(w)=w_1w_3, \\[5pt]
R_5(u)=-v_2w_4, \quad R_5(v)=v_2w_3,\quad R_5(w)=w_2w_3, \\[5pt]
R_6(u)=-w_5-(v_4+w)w_4-u_4w_3, \quad R_6(v)=w_6+(u_3+w)w_3+v_3w_4, \quad R_6(w)=0,
  \end{array}
  \end{equation}
respectively. Equations (\ref{triple2}) must be considered together with equations (\ref{alttriple}) that come from the involutivity conditions of the first three vector fields, $Y_1, Y_2, Y_3$. The characteristic variety of the full system (\ref{alttriple})+(\ref{triple2}) is the triple scroll $S_{2,2}$, where the scroll $S_{2,2}$  is defined by the equations
  \begin{equation}\label{p22}
  \begin{array}{c}
u_2p_2p_3+v_1p_2p_4-u_2p_1p_3-v_2p_1p_4=0,\\[3pt]
u_1p_3p_4+v_1p_4^2-p_1p_5-u_4p_1p_3-(v_4+w)p_1p_4=0,\\[3pt]
u_2p_3p_4+v_2p_4^2-p_2p_5-u_4p_2p_3-(v_4+w)p_2p_4=0,\\[3pt]
v_1p_3p_4+u_1p_3^2-p_1p_6-v_3p_1p_4-(u_3+w)p_1p_3=0,\\[3pt]
v_2p_3p_4+u_2p_3^2-p_2p_6-v_3p_2p_4-(u_3+w)p_2p_3=0,\\[3pt]
p_3p_5-p_4p_6+u_4p_3^2-v_3p_4^2+(v_4-u_3)p_3p_4=0.
  \end{array}
  \end{equation}
It can be parametrised as
$p_i dx^i=\alpha\, \omega(\lambda)+\beta\, \phi(\lambda)$
where $\alpha, \beta$ are arbitrary functions and $\omega(\lambda), \phi(\lambda)$ are 1-forms depending on auxiliary parameter $\lambda$:
  \begin{equation}\label{omega222}
  \begin{array}{c}
\omega(\lambda)=u_1dx^1+u_2dx^2+\lambda(dx^3-u_4dx^5-(u_3+w)dx^6)+\lambda^2dx^6,\\[5pt]
\phi(\lambda)=v_1dx^1+v_2 dx^2+\lambda(dx^4-(v_4+w)dx^5-v_3dx^6)+\lambda^2dx^5.
  \end{array}
  \end{equation}
The Lax distribution $\D=\langle Y_1, Y_2, Y_3, Y_4\rangle$ satisfies the equations $\omega(\lambda)=\phi(\lambda)=0$.

It will be shown in Section \ref{sec:alpha-par22} that system (\ref{alttriple})+(\ref{triple2}) gives a local parametrisation of {\it all} involutive $S_{2,2}$ scroll structures. A simple parameter count (12 arbitrary functions of 3 variables) suggests that this parametrisation must be equivalent to that obtained in Section \ref{s22c1}, however, direct transformation establishing equivalence of systems (\ref{triple})+(\ref{triple1}) and (\ref{alttriple})+(\ref{triple2})  is a nonlocal (B\"acklund-type) transformation.

\section{General involutive scroll structures}
\label{sec:gen}

In this section we demonstrate that  involutive scroll structures are governed by a
dispersionless integrable system and derive the corresponding Lax representation. 
The general construction is as follows. 
Given an involutive  scroll structure (\ref{scroll}), one looks for foliations of $M^{n+1}$ by $\alpha$-manifolds,
assuming that $\lambda$ is a function on $M$ and imposing the relations
 \begin{equation}\label{inv}
(d\omega(\lambda) +d\lambda\wedge \omega'(\lambda))\wedge \omega(\lambda) \wedge \phi(\lambda)=0, \qquad (d\phi(\lambda)+d\lambda \wedge \phi'(\lambda))\wedge \omega(\lambda) \wedge \phi(\lambda)=0,
 \end{equation}
where we use the notation
$d \omega (\lambda)=d \omega_0+\lambda d \omega_1+ \dots + \lambda^{k}d \omega_{k}, \ \omega'(\lambda)=\frac{d}{d\lambda}\omega(\lambda)$,
etc. Eliminating $d\lambda$-terms from (\ref{inv}), we obtain an exterior differential system
 \begin{equation}\label{int}
 \begin{array}{c}
d \omega(\lambda) \wedge \omega'(\lambda)\wedge \omega(\lambda) \wedge \phi(\lambda)=0,\\[7pt]
d \phi(\lambda) \wedge \phi'(\lambda)\wedge \omega(\lambda) \wedge \phi(\lambda)=0,\\[7pt]
d \omega(\lambda) \wedge \phi'(\lambda)\wedge \omega(\lambda) \wedge \phi(\lambda)+
d \phi(\lambda) \wedge \omega'(\lambda)\wedge \omega(\lambda) \wedge \phi(\lambda)=0,
 \end{array}
 \end{equation}
which is to  be split in powers of $\lambda$. Compatibility conditions of equations (\ref{inv}) for $\lambda$, 
along with equations (\ref{int}), can be viewed as a differential system $\Sigma$ governing involutive scroll structures. 

Relations \eqref{inv} play the role of a dispersionless Lax representation of $\Sigma$. Indeed, 
tangents to $\alpha$-manifolds in $M^{n+1}$ are encoded by spans $\langle V_i^0\rangle$ of $(n-1)$ vectors 
annihilated by $\omega,\phi$ and depending on $\lambda$. Asuming $\lambda=\lambda(x)$ to be a
function on $M$ makes $V^0_i$ into vector fields on $M$, then
the conditions $\omega([V^0_i,V^0_j])=0$, $\phi([V^0_i,V^0_j])=0$ are equivalent to involutivity.

Simulteneously, one can encode $\lambda$-freedom of $V^0_i$ at $x\in M$ as a bundle $\hat{M}^{n+2}\to M^{n+1}$,
with $\mathbb{P}^1(\lambda)$ fiber, called the correspondence space. $V^0_i$ lift to vector fields $V_i$ 
on $\hat{M}$ (via first-order PDEs in the commutativity conditions for $V^0_i$),
generating distribution $\D$ of rank $(n-1)$, which is the dispersionless Lax representation. 
Its Frobenius condition is equivalent to \eqref{inv}. 
 
This interpretation means the characteristic property of the dispersionless Lax reprsentation. For Lax pairs and 
determined PDEs it was demonstrated in full generality in \cite{CK}. Hierarchies consitute overdetermined systems and
the rank of Lax distribution is larger than two, henceforth, we observe this property  aposteriori.

Below, we provide an explicit coordinate form of equations (\ref{inv}), (\ref{int})  in the special case of involutive 
$S_{1,2}$  and $S_{2,2}$  scroll structures. 
Remarkably, the general cases are governed by the hierarchy of conformal self-duality equations.

\subsection{Parametrisation of involutive $S_{1,2}$ scroll structures: proof of Theorem \ref{t12}} 
\label{sec:alpha-par}

The general $S_{1,2}$ scroll structure in 5D has the form
 $$
\omega(\lambda)=\omega_0+\lambda \omega_1, \qquad
\phi (\lambda)=\phi_0+\lambda \phi_1+ \lambda^{2}\phi_{2},
 $$
for some generic 1-forms $\omega_i$, $\phi_j$ on the manifold $M^5$. We now explore the condition of involutivity.
 
Using linear fractional transformation freedom for $\lambda$, we can assume that the distribution $\omega(\lambda)=\phi(\lambda)=0$ is integrable for $\lambda=0$ and $\lambda=\infty$.
Thus, there are coordinates $x^1, x^2$ and $x^3, x^5$ such that $\omega_0=adx^1+bdx^2, \ \phi_0=\tilde a dx^1+\tilde b dx^2$ and $\omega_1=qdx^3+pdx^5, \ \phi_2=\tilde q dx^3+\tilde p dx^5$. Applying the transformation $\phi(\lambda)\to \phi(\lambda)-\lambda \frac{\tilde q}{q}\omega(\lambda)$, one can set $\tilde q=0$. Using the scaling freedom of  $\omega(\lambda)$ and $\phi(\lambda)$, one can set  $q=\tilde p=1$. Thus,
 $$
\omega(\lambda)=(adx^1+bdx^2)+\lambda (dx^3+pdx^5), \qquad
\phi (\lambda)=(\tilde a dx^1+\tilde b dx^2)+\lambda \phi_1+ \lambda^{2}dx^5.
 $$
Equating to zero the $\lambda^4$-term in equation $(\ref{int})_2$ gives 
 $$
d\phi_1\wedge \phi_1\wedge dx^3\wedge dx^5=0,
 $$
which means that the distribution $\phi_1=dx^3=dx^5=0$ is integrable. 
Thus, one can find a coordinate $x^4$ such that
 \begin{equation*}\label{phi}
\phi_1=r dx^4+s dx^3+q dx^5.
 \end{equation*}
Applying the transformation $\phi(\lambda)\to \phi(\lambda)-s\omega(\lambda)$, one can set  $s=0$. 
Finally, using the scaling $\lambda \to r \lambda, \ \omega \to r \omega, \ \phi \to r^2 \phi$, one can set $r=1$.
Ultimately,
 \begin{equation}\label{alt}
 \begin{array}{c}
\omega(\lambda)=(\alpha dx^1+\beta dx^2)+\lambda (dx^3+pdx^5), \\[5pt]
\phi (\lambda)=(\gamma dx^1+\delta dx^2)+\lambda (dx^4+qdx^5)+ \lambda^{2}dx^5,
 \end{array}
 \end{equation}
compare with (\ref{altomega12}).  Here $\alpha, \beta, \gamma, \delta, p, q$ are some functions yet to be specified. 

At this point the remaining conditions (\ref{int}) are not sufficient to proceed, and we turn to
representation (\ref{inv}) or even to a more convenient dispersionless Lax distribution $\D$, which is 
the kernel of $\omega(\lambda),\phi(\lambda)$ lifted to the correspondence space $\hat{M}^6$. The latter is 
the natural $\mathbb{P}^1$ bundle (with fiber coordinate $\lambda$) over $M^5$ with a basis of generators
 \begin{gather*}
V_1=\lambda\p_1 - \a\p_3 - \ga\p_4 + S_1\p_\lambda,\quad
V_2=\lambda\p_2 - \b\p_3 - \d\p_4 + S_2\p_\lambda,\\[5pt]
V_3=\p_5 - p\p_3 - (q+\lambda)\p_4 + S_3\p_\lambda,
 \end{gather*}
where the coefficients $S_1, S_2, S_3$ are some functions on $\hat{M}^6$ to be specified later. The Frobenius integrability conditions for $\D=\langle V_1,V_2,V_3\rangle$ contain the following short equations:
 $$
\a(\ga_2-\d_1)=\ga(\a_2-\b_1),\quad \b(\ga_2-\d_1)=\d(\a_2-\b_1).
 $$
To ensure nondegeneracy of the scroll we assume the determinant $\a\d-\b\ga\neq0$ and hence we can potentiate
$\a=u_1$, $\b=u_2$, $\ga=v_1$, $\d=v_2$ for some functions $u,v$ on $M^5$.
Then another short equation becomes $u_1(p_2+u_{24})=u_2(p_1+u_{14})$, whence
$p=-u_4-f(u, x^3, x^4, x^5)$.

Moreover, the same integrability conditions give the form of the coefficients $S_i$, namely, each of them is
$\lambda$ times a function on $M^5$ except for $S_3$, that can have an additional term with $\lambda^2$.
Substituting this into the dispersionless Lax triple we get refined expressions
 \begin{gather*}
V_1=\lambda\p_1 - u_1\p_3 - v_1\p_4 + \lambda w_a\p_\lambda,\quad
V_2=\lambda\p_2 - u_2\p_3 - v_2\p_4 + \lambda w_b\p_\lambda,\\[5pt]
V_3=\p_5 + (u_4+f)\p_3 - (q+\lambda)\p_4 + (\lambda w_c+\lambda^2h)\p_\lambda,
 \end{gather*}
where the function $f$ is as above, while $q,w_a,w_b,w_c,h$ on $M^5$ are yet to be specified.
The simplest of the Frobenius conditions for $\D$ are as follows:
 $$
h_1=h_2=0,\ f_u=h,\ (w_a)_2=(w_b)_1,
 $$
as well as a few equations on $q$ and $w_c$. 
Thus, $w_a=w_1$, $w_b=w_2$ for a function $w$ on $M^5$, $h=h(x^3,x^4,x^5)$ and $f$ is affine in $u$.

Next, substitution $f=h\,u+m(x^3,x^4,x^5)$, $q=n-hv-v_4-w$ simplifies those few equations above to
$n_1=n_2=0$, so $n=n(x^3,x^4,x^5)$, and 
$w_c=\lambda(h_3u+h_4v-2hw-w_4+r)+\lambda^2h$ for $r=r(x^3,x^4,x^5)$.
The functions $m,n,r$ can be eliminated via a change $u\mapsto u+\mu$, $v\mapsto v+\nu$, $w\mapsto w+\rho$,
where $\mu=\mu(x^3,x^4,x^5)$, $\nu=\nu(x^3,x^4,x^5)$, $\rho=\rho(x^3,x^4,x^5)$ satisfy some simple
first-order differential equations. 

Finally, we get the dispersionless Lax representation
 \begin{equation}\label{Pi5D}
 \begin{array}{c}
V_1=\lambda\p_1 - u_1\p_3 - v_1\p_4 + \lambda w_1\p_\lambda,\quad
V_2=\lambda\p_2 - u_2\p_3 - v_2\p_4 + \lambda w_2\p_\lambda,\\[5pt]
V_3=\p_5 + (u_4+hu)\p_3 + (v_4+hv+w-\lambda)\p_4 + \lambda(h_3u+h_4v-2hw-w_4+\lambda h)\p_\lambda,
 \end{array}
 \end{equation}
in which $h=h(x^3,x^4,x^5)$ is a functional parameter that may manifest an integrable deformation.
And indeed, it is an auxiliary arbitrary function, which enters all compatibility equations, passes the involutivity test 
and governs a deformation of the scroll.
However, this parameter $h$ is removable. To see this let us return to the equation of ISS in terms of differential forms
$\omega(\lambda)$, $\phi(\lambda)$ whose expressions are now specified to
 \begin{equation}\label{hdef}
 \begin{array}{c}
\omega(\lambda)=u_1dx^1+u_2dx^2+\lambda(dx^3-(u_4+hu)dx^5),\\[7pt]
\phi(\lambda)=v_1dx^1+v_2dx^2+\lambda(dx^4-(v_4+hv+w)dx^5)+\lambda^2dx^5.
 \end{array}
 \end{equation}
The parameter $h=h(x^3, x^4, x^5)$ can be eliminated by introducing  $f=f(x^3, x^4, x^5)$ such that 
$f_{44}=hf_4$, followed by a  point transformation
 \begin{gather*}
\tilde x^i= x^i\ (i\neq4),\ \  \tilde x^4= f(x^3, x^4, x^5), \ \
\tilde u= f_4u,\ \  \tilde v= f_4^2v+f_3f_4u,\\ 
\tilde w= f_4w-f_{44}v-f_{34}u+f_5,\ \  \tilde \lambda = f_4\lambda.
 \end{gather*}
In the variables with tildas, the forms $\tilde \omega=f_4\, \omega(\lambda)$ and $\tilde \phi=f_4^2\, \phi(\lambda)+f_3f_4\, \omega(\lambda)$  take canonical form (\ref{altomega12}). 
Once the forms $\omega(\lambda)$ and $\phi(\lambda)$ are in  canonical form (\ref{altomega12}),  conditions (\ref{inv}), (\ref{int}) lead to PDEs (\ref{alttriple}). These in turn have Lax representation 
(\ref{altLaxh}) that is precisely \eqref{Pi5D} with $h=1$.
Once we have aligned the parametrization with that of construction from Section \ref{s12c2}, the remaining claims
follow. This finishes the proof of Theorem \ref{t12}.

\subsection{Parametrisation of involutive $S_{2,2}$ scroll structures: proof of Theorem \ref{t22}} 
\label{sec:alpha-par22}

The general $S_{2,2}$ scroll structure in 6D has the form
  $$
\omega(\lambda)=\omega_0+\lambda \omega_1+\lambda^2 \omega_2, \qquad
\phi (\lambda)=\phi_0+\lambda \phi_1+ \lambda^{2}\phi_{2},
  $$
for some generic 1-forms $\omega_i$, $\phi_j$ on the manifold $M^6$. We now explore the condition of involutivity.

Using linear fractional transformation freedom for $\lambda$, we can assume that the distribution 
$\omega(\lambda)=\phi(\lambda)=0$ is integrable for $\lambda=0$ and $\lambda=\infty$.
Thus, there are coordinates $x^1, x^2$ and $x^5, x^6$ such that 
$\omega_0=adx^1+bdx^2, \ \phi_0=\tilde a dx^1+\tilde b dx^2$ and 
$\omega_2=pdx^5+qdx^6, \ \phi_2=\tilde p dx^5+\tilde q dx^6$. 
Taking linear combinations of $\omega(\lambda)$ and $\phi(\lambda)$, one can set 
$\omega_2=dx^6, \ \phi_2=dx^5$. Thus,
  $$
\omega(\lambda)=\alpha dx^1+\beta dx^2+\lambda \omega_1+\lambda^2dx^6, \qquad
\phi (\lambda)=\gamma dx^1+\delta dx^2+\lambda \phi_1+ \lambda^{2}dx^5.
  $$
At the $\lambda^5$-terms in  (\ref{int}), one gets
  $$
d\omega_1\wedge \omega_1\wedge dx^5\wedge dx^6=0, \quad 
d\phi_1\wedge \phi_1\wedge dx^5\wedge dx^6=0, \quad 
(d\omega_1\wedge \phi_1+d\phi_1\wedge \omega_1)\wedge dx^5\wedge dx^6=0, 
  $$
respectively, implying that the three-dimensional distributions defined by the equations $\omega_1= dx^5= dx^6=0$ and  $\phi_1= dx^5= dx^6=0$, are integrable. Thus,  there exist coordinates $x^3, x^4$ such that 
  $$
\omega_1=pdx^3+\mu dx^5+\nu dx^6, \quad \phi_1=qdx^4+\tau dx^5+\eta dx^6,
  $$
furthermore, $h=q/p$ must be a function of $x^3, x^4, x^5, x^6$ only. By suitable rescalings of 
$\lambda, \omega(\lambda), \phi(\lambda)$, one can set $p=1$. On redefinition of the coefficients, 
one ultimately gets
  \begin{equation}\label{omla6}
\begin{array}{c}
\omega(\lambda)=\alpha dx^1+\beta dx^2+\lambda (dx^3+\mu dx^5+\nu dx^6) +\lambda^2dx^6, \\[5pt]
\phi (\lambda)=\gamma dx^1+\delta dx^2+\lambda (h dx^4+\tau dx^5+\eta dx^6)+ \lambda^{2}dx^5.
\end{array}
  \end{equation}
  
Similar to the proof of Theorem \ref{t12} of the previous section, to proceed one needs the full set of equations, 
given through the dispersionless Lax representation,
which is a distribution $\D$ on $\hat{M}^7$ lifting the kernel of $\omega(\lambda),\phi(\lambda)$
on $M^6$. Generators of $\D$ are
 \begin{gather*}
V_1=\lambda\p_1 - \a\p_3 - h^{-1}\ga\p_4 + S_1\p_\lambda,\quad
V_3=\p_5 - m\p_3 - h^{-1}(p+\lambda)\p_4 + S_3\p_\lambda, \\[5pt]
V_2=\lambda\p_2 - \b\p_3 - h^{-1}\d\p_4 + S_2\p_\lambda,\ \quad 
V_4=\p_6 - (n+\lambda)\p_3 - h^{-1}q\p_4 + S_4\p_\lambda,
 \end{gather*}
where the coefficients $S_1, S_2, S_3, S_4$ are some functions on $\hat{M}^7$ to be specified later. The Frobenius integrability conditions for $\D=\langle V_1,V_2,V_3,V_4\rangle$ contain the following short equations:
 $$
h(\a_2-\b_1)=h(\ga_2-\d_1)=0,\quad h_1=h_2=0,\quad hm_1+a_4=hm_2+b_4=0,
 $$
from which we get $\a=u_1$, $\b=u_2$, $\ga=v_1$, $\d=v_2$ for some functions $u,v$ on $M^6$,
as well as $m=-h^{-1}u_4+f(x^3,x^4,x^5,x^6)$ and $h=h(x^3,x^4,x^5,x^6)$.
Furthermore, we obtain the equations 
 $$ 
\ga(\d_3+p_2)=\d(\ga_3+p_1),\quad 2\d h_3=h(\d_3+p_2),
 $$
from which we first get $p=\tilde{g}(v,x^3,x^4,x^5,x^6)-v_3$ and then
$p=2h^{-1}h_3v-v_3+g(x^3,x^4,x^5,x^6)$.
Note that a change $u\mapsto u+\mu$ and $v\mapsto v+\nu$, where the functions $\mu=\mu(x^3,x^4,x^5,x^6)$ 
and $\nu=\nu(x^3,x^4,x^5,x^6)$ satisfy some simple first-order differential equations,
eliminates $f$ and $g$. 

Moreover, the same integrability conditions give the form of the coefficients $S_i$, namely, 
$S_1$ and $S_2$ are affine in $\lambda$, while $S_3$ and $S_4$ are quadratic in $\lambda$ with zero free term.
Substituting this into the dispersionless Lax quadruple we obtain refined expressions,
 \begin{equation*}
 \begin{array}{l}
V_1=\lambda\p_1 - u_1\p_3 - h^{-1}v_1\p_4 + (w_a\lambda+z_a)\p_\lambda,\\[5pt]
V_2=\lambda\p_2 - u_2\p_3 - h^{-1}v_2\p_4 + (w_b\lambda+z_b)\p_\lambda,\\[5pt] 
V_3=\p_5 + h^{-1}u_4\p_3 - h^{-1}(p+\lambda)\p_4 + (w_c\lambda+z_c\lambda^2)\p_\lambda, \\[5pt]
V_4=\p_6 - (n+\lambda)\p_3 + h^{-1}(v_3-2h^{-1}h_3v)\p_4 + (w_d\lambda+z_d\lambda^2)\p_\lambda,
 \end{array}
 \end{equation*}
where the function $h$ is as above, while the functions $p,n,w_a,w_b,w_c,w_d,z_a,z_b,z_c,z_d$ on $M^6$ 
are yet to be specified.
The Frobenius conditions for $\D$ yield the following particular equations:
 \begin{gather*}
(w_b)_1=(w_a)_2,\ (z_c)_1=(z_c)_2=0,\ (z_d)_1=(z_d)_2=0, \\[2pt]
h (w_c)_1 + (w_a)_4 = h u_1 (z_c)_3 + v_1 (z_c)_4 - 2h w_a z_c ,\
h (w_c)_2 + (w_b)_4 = h u_2 (z_c)_3 + v_2 (z_c)_4 - 2h w_b z_c, \\[2pt]
(w_a)_3 + (w_d)_1 = u_1 (z_d)_3 + h^{-1}v_1 (z_d)_4- 2 w_a z_d,\,
(w_b)_3 + (w_d)_2 = u_2 (z_d)_3 + h^{-1}v_2 (z_d)_4- 2 w_b z_d.
 \end{gather*}
This implies $w_a=w_1$, $w_b=w_2$ for a function $w$ on $M^6$ and
 $$
w_c=-h^{-1}w_4-2z_cw+(z_c)_3u+(z_c)_4h^{-1}v+\tilde{t},\quad
w_d=-w_3-2z_dw+(z_d)_3u+(z_d)_4h^{-1}v+t
 $$ 
for some functions $\tilde{t},t,z_c,z_d$ of $(x^3,x^4,x^5,x^6)$. A shift $w\mapsto w+\rho$ eliminates one of
the functions $\tilde{t},t$, so we set $\tilde{t}=0$ in what follows. 

Further exploration of the integrability conditions gives $u_2z_a-u_1z_b=0$, $v_2z_a-v_1z_b=0$,
and again using that $\a\d-\b\ga\neq0$ we conclude $z_a=z_b=0$. 
Substituting this into the dispersionless Lax representation and recomputing the Frobenius conditions we get 
$z_c=0$, $z_d=-h^{-1}h_3$, as well as the following simple equations:
 \begin{gather*}
(z_d)_4=0,\quad
n_1 + u_{13} + z_d u_1 + w_1 =0,\ \ n_2 + u_{23} + z_d u_2 + w_2 =0,\\[2pt]
p_1 + h^{-1}h_3 u_1 + h^{-1}v_{14} + w_1 =0,\ \ p_2 + h^{-1}h_3 u_2 + h^{-1}v_{24} + w_2 =0.
 \end{gather*}
Whence, $n=-u_3-h^{-1}h_3u-w+r$, $p=-h^{-1}h_3u-h^{-1}v-w+s$, where $r,s$ are functions of 
$(x^3,x^4,x^5,x^6)$. Now the Lax quadruple takes the form
 \begin{equation*}
 \begin{array}{rl}
V_1 = &\!\!\! \lambda\p_1 - u_1\p_3 - h^{-1}v_1\p_4 + w_1\lambda\p_\lambda,\\[5pt]
V_2 = &\!\!\! \lambda\p_2 - u_2\p_3 - h^{-1}v_2\p_4 + w_2\lambda\p_\lambda,\\[5pt] 
V_3 = &\!\!\! \p_5 + h^{-1}u_4\p_3 + \bigl(h^{-1}(h^{-1}v_4+h^{-1}h_3u+w-\lambda)-s\bigr)\p_4 - h^{-1}
 w_4\lambda\p_\lambda, \\[5pt]
V_4 = &\!\!\! \p_6 + (u_3+h^{-1}h_3u+w-\lambda-r)\p_3 + h^{-1}(v_3-2h^{-1}h_3v)\p_4 \\[2pt]
      &\! -\bigl((h^{-1}h_3)_3 u + h^{-1}(h^{-1}h_3)_4 v +w_3 - 2h^{-1}h_3w-t+ h^{-1}h_3\lambda\bigr)\lambda\p_\lambda.
 \end{array}
 \end{equation*}
One of the simplest equations in the Frobenius conditions we find is $hh_{34} = h_3h_4$, and its substitution 
into other equations gives $h_3=0$. Thus, $h=h(x^4,x^5,x^6)$, and a change of $x^4$ in formula
\eqref{omla6} simplifies $h$ (also changing $\tau,\eta$ but this is not essential) so we can assume $h=1$. 
We run again the same line of computations with $h=1$ in the last Lax representation, which gives
$r_4=0$, $s_3=-t$, $t_4=0$. Thus, we can set $r=\p_3\mu(x^3,x^5,x^6)$, 
$s=\p_4\rho(x^4,x^5,x^6)-\nu(x^3,x^5,x^6)$, $t=\p_3\nu(x^3,x^5,x^6)$.
Then the change of dependent variables $u\mapsto u+\mu$, $v\mapsto v+\rho$, $w\mapsto w+\nu$ eliminates
$r,s,t$. 
With these last parameters normalized, our Lax distribution becomes identical to \eqref{altLaxh22} and the remaining 
claims follow from the construction of Section \ref{s22c2}. This finishes the proof of Theorem \ref{t22}.

\section{Involutive scroll structures in arbitrary dimensions}
\label{SIH}

In Section \ref{sec:ex} we discussed truncated hierarchies of heavenly type equations.
Here we consider the full hierarchy for the second Plebansky equation, as well as justify the claims from the
introduction concerning the general heavenly hierarchy. Other hierarchies are expected to lead to 
similar formulae. Note that since heavenly equations are related by a B\"acklund type transformation,
relations between the hierarchies (and  the corresponding ISS) are nonlocal.

\subsection{Scrolls $S_{k-1,k}$ and $S_{k,k}$ for the second heavenly hierarchy}

We are going to extends the results of Section \ref{s12c1}(b) to arbitrary dimension. 
The hierarchy of the second Plebansky equation is
 \begin{equation}\label{2PbHier}
u_{a,b+2}-u_{a,b+2}+u_{1,a}u_{2,b}-u_{1,b}u_{2,a}=0\qquad (0<a<b<\infty).
 \end{equation}
 
 \begin{theorem}\label{ISS2PBE}
System \eqref{2PbHier} is involutive and has involutive scroll structure on every solution.
Truncation of the hierarchy to $n+1$ dimensions yields a family of ISS depending on $n-1$ functions 
of 3 variables.
 \end{theorem}
  
Involutivity of this system is a straightforward verification (commutativity of flows of the hierarchy). 
Let us work in un-restricted dimension,
encoded as a manifold $M^\infty$ (inductive limit of finite-dimensional manifolds $M^n$).
The correspondence space $\hat{M}^{1+\infty}$ is a $\mathbb{P}^1$ bundle over it, supplied with
a canonical Lax distribution $\D^\infty$ generated by  vector fields
 \begin{equation*}
 \begin{array}{l}
V_1=\p_3+(u_{12}-\lambda)\p_1-u_{11}\p_2,\\[3pt] 
V_2=\p_4+u_{22}\p_1-(u_{12}+\lambda)\p_2,\\[3pt]
V_3=\p_5+(u_{23}+\lambda u_{12}-\lambda^2)\p_1-(u_{13}+\lambda u_{11})\p_2,\\[3pt] 
V_4=\p_6+(u_{24}+\lambda u_{22})\p_1-(u_{14}+\lambda u_{12}+\lambda^2)\p_2,\\[3pt]
V_5=\p_7+(u_{25}+\lambda u_{23}+\lambda^2 u_{12}-\lambda^3)\p_1-(u_{15}+\lambda u_{13}+\lambda^2 u_{11})\p_2,\\[3pt] 
V_6=\p_8+(u_{26}+\lambda u_{24}+\lambda^2 u_{22})\p_1-(u_{16}+\lambda u_{14}+\lambda^2 u_{12}+\lambda^3)\p_2,\\[3pt]
\qquad\dots\qquad\dots\qquad\dots\qquad\dots
 \end{array} 
 \end{equation*}
Note 2-periodicity in the pattern of this extension. Alternatively, $\D^\infty$ 
is generated by vector fields
 $$
X_j=\p_{j+2}+u_{2,j}\p_1-u_{1,j}\p_2-\lambda\p_j, \qquad (0<j<\infty),
 $$
and the Frobenius condition for this distribution is satisfied modulo \eqref{2PbHier}. To be more precise, 
the coefficients of the commutators $[V_a,V_b]$ are third-order differential polynomials that are corollaries of second-order expressions \eqref{2PbHier}. 

The kernel of the pushforward of $\D^\infty$ to $M^\infty$ is annihilated by
 \begin{equation*}
 \begin{array}{l}
\omega(\lambda) = \displaystyle
	\sum_{i=1}^\infty\Bigl(\lambda^{i-1}+\sum_{j=0}^{i-2}u_{1,2i-2j-2}\lambda^j\Bigr) dx^{2i}
	+\sum_{i=1}^\infty\Bigl(\sum_{j=0}^{i-1}u_{1,2i-2j-1}\lambda^j\Bigr) dx^{2i+1}\Bigr),\\[9pt] 
\phi(\lambda) = \displaystyle
	\sum_{i=0}^\infty\Bigl(\lambda^{i}+\sum_{j=0}^{i-1}u_{2,2i-2j-1}\lambda^j\Bigr) dx^{2i+1}
	+\sum_{i=1}^\infty\Bigl(\sum_{j=0}^{i-1}u_{2,2i-2j}\lambda^j\Bigr) dx^{2i+2}\Bigr).
 \end{array} 
 \end{equation*}
In this form, one can easily see that the truncation of $M^\infty$ to $M^{2k+1}(x^1,\dots,x^{2k+1})$
yields the scroll $S_{k-1,k}$, while for $M^{2k+2}(x^1,\dots,x^{2k+2})$ we get $S_{k,k}$. 
These scroll structures are involutive by the construction of the hierarchy ($\a$-manifolds are
integral leafs of the distribution $\D$).
Note that the directrix of $S_{k-1,k}$ is $\mathbb{P}^{k-1}$ spanned by the forms 
 $$
\psi_j= dx^{2j+2}+\sum_{i=j+2}^k u_{1,2i-2j-2} dx^{2i}+\sum_{i+1}^k u_{1,2i-2j-1} dx^{2i+1}\quad
(0\leq j<k),
 $$
(the scrolls $S_{k,k}$ have no directrix.) 
 
Since $M^\infty$ is filtered as $\cdots\subset M^n\subset M^{n+1}\subset\cdots$, we have the
filtration of its tangent bundle $\cdots\subset TM^n\subset TM^{n+1}\subset\cdots$, while the 
cotangent bundle is obtained by the projective limit 
$\cdots\leftarrow T^*M^n\leftarrow T^*M^{n+1}\leftarrow \cdots$.
The projective limit $S_{\infty,\infty}\subset T^*M^\infty$ of the scrolls 
 $$
\cdots\leftarrow S_{k-1,k}\leftarrow S_{k,k}\leftarrow S_{k,k+1}\leftarrow\cdots
 $$ 
should be considered as the characteristic variety of the full unrestricted hierarchy.

Due to involutivity, the functional freedom of solutions $u$ of the truncated system is given by the degree and 
dimension of the corresponding characteristic variety \cite{BCG, KL}. Since the latter is a scroll, the local
functional count is as stated. This finishes the proof of Theorem \ref{ISS2PBE}.

\subsection{Scrolls $S_{k-1,k}$ and $S_{k,k}$ for the general heavenly hierarchy}
\label{scrGHE}

Let us give some more details on the proof of Theorem \ref{ISSGPBE} on ISS for system \eqref{gh}, 
which was sketched in Section \ref{Sghe}.
The involutivity of the system follows from commutativity of the flows of the hierarchy.
Thus, the functional freedom of solutions $u$ of the truncated system is given by the degree and 
dimension of the corresponding characteristic variety \cite{BCG, KL}. 
Yet, contrary to the previous section, it is not straightforward to observe the scroll structure 
for the general heavenly hierarchy. 

Actually, while the $[\a:\b]$ lines $\mathbb{P}^1$ of \eqref{rulingab} give ruling of the surface,
the rational normal curves parametrized by vectors $r$ and $s$ in this formula are not supported on 
non-intersecting projective subspaces as in the definition of scrolls.
However, a straightforward verification shows that the surface given by parametrization \eqref{rulingab} 
is smooth (beware: the parametrization is not smooth on the entire projective surface).
Thus, from del Pezzo's theorem (recalled in Section \ref{RNS0}) we conclude that the surface is indeed a scroll,
that is, one of $S_{k,l}$.

Furthermore, direct computation shows that the symmetry algebra of surfaces \eqref{rulingab} is 6-dimensional, 
with the Lie algebra structure alternating between (semi-simple) and (semi-direct product), according to
the parity of $n$. For $n$ even and odd, they are isomorphic to
 $$
\mathfrak{aff}_2\simeq\mathfrak{gl}_2\ltimes\R^2\quad\text{and}\quad
\mathfrak{so}_{2,2}\simeq\mathfrak{sl}_2\oplus\mathfrak{sl}_2,
 $$
respectively. These symmetry algebras, as noticed in Section \ref{RNS0}, uniquely determine type of the scroll.
Computation of symmetries is discussed in the next Section \ref{sec:scrollsym}.

Thus, showing that the degree of the surface is $n-1$ and using its nondegeneracy (meaning its secants span 
the entire space), we conclude that the variety $\op{Char}(\E)$ for system $\E$ given by \eqref{gh}
is a scroll $S_{k-1,k}$ for $n=2k$ and $S_{k,k}$ for $n=2k+1$, as stated before Theorem \ref{ISSGPBE}.
Since $\op{Char}(\E)$ has degree $n-1$ and (affine) dimension 3, the general solution of system \eqref{gh} 
depends on $n-1$ arbitrary functions of 3 variables \cite{BCG, KL}. This finsihes the proof of Theorem \ref{ISSGPBE}.

\medskip

The proof identifies scrolls indirectly. One could wonder what is the directrix of the scroll $S_{k-1,k}$
(for $S_{k,k}$ it is not defined). Explicit formulae here are involved, so we give the following description.

The scroll $S_{k-1,k}$ is  a ruled surface, i.e.,\ a one-parametric family of lines $L_\lambda$ in 
$\mathbb{P}^{2k}$ where  $\lambda$ is the projective parameter. In the affine version, lines $\mathbb{P}^1$ 
correspond to planes $\Pi^2$. For general heavenly equation they have the form 
$\Pi_\lambda=(A-\lambda)^{-1}\Pi_\infty$, where $\Pi_\infty=\langle r_i, s_i\rangle$ is a fixed plane
and $A=\op{diag}\bigl(a_i\bigr)$ is a linear operator.

Fix $(k+1)$ values $\lambda_1,\dots,\lambda_{k+1}\in\mathbb{P}^1$ and let $\Pi_i=\Pi_{\lambda_i}$ be the 
corresponding 2-planes. By nondegeneracy, they are in general position. 
Let $U_i=\oplus_{j\neq i}\Pi_i$ be $2k$-subspaces  labelled by $i=1,\dots,2k+1$,
and let $\ell_i=U_i\cap\Pi_i$ be the intersection lines. They correspond to points in $\mathbb{P}^{2k}$
and any $k$-tuple of points in general position uniquely determines the subspace $\mathbb{P}^{k-1}$ through them.

 \begin{lemma}
 \label{l1}
Points $[\ell_1],\dots,[\ell_{k+1}]$ all belong to some $\mathbb{P}^{k-1}$.
 \end{lemma}
 
 \begin{proof}
One can always choose coordinates such that 
 $$
\Pi_1=\langle e_1,e_2\rangle, \dots, \Pi_k=\langle e_{2k-1},e_{2k}\rangle,\
\Pi_{k+1}=\Bigl\langle e_{2k+1},\sum_{i=1}^{2k+1}x^ie_i\Bigr\rangle.
 $$
Then $\ell_1=\langle x^1e_1+x^2e_2\rangle$, \dots, $\ell_k=\langle x^{2k-1}e_{2k-1}+x^{2k}e_{2k}\rangle$,
$\ell_{k+1}=\langle \sum_{i=1}^{2k}x^ie_i\rangle$, which all belong to $\mathbb{P}^{k-1}$ generated by $\ell_1, \dots, \ell_k$.
 \end{proof}


 \begin{prop}
The directrix of the scroll $S_{k-1,k}$ lies in the subspace $\mathbb{P}^{k-1}$ from Lemma \ref{l1}.
 \end{prop}
 
 \begin{proof}
Let $V=\sum\ell_i$ be the linear space corresponding to the projective subspace $\mathbb{P}^{k-1}$ in question.
Let $v_1,\dots,v_k$ be its basis and $\alpha_1,\dots,\alpha_{2k-1}$ be a basis of $\Pi_\infty^\perp$. 
Then the condition $ V  \cap\Pi_\lambda\neq0$, which is equivalent to $(A-\lambda)V\cap\Pi_\infty\neq0$, means that the $k\times(2k-1)$ matrix $B$ with entries
$b_{ij}=\alpha_j\bigl((A-\lambda)v_i\bigr)$ has rank $<k$. This is equivalent to the vanishing of all
$k\times k$ minors of $B$, thus given by an ideal generated by polynomials of degree $k$ in $\lambda$.
These polynomials vanish for different $\lambda_1,\dots,\lambda_{k+1}$, and hence vanish identically.
 \end{proof}

Let us discuss another approach to finding  directrix of $S_{k,l}$ for $k<l$. Each ruling $\mathbb{P}^1$ of 
$S_{k,l}\subset\mathbb{P}^n$ must intersect a fixed projective space $\mathbb{P}^k$.
This imposes a system of linear relations on its Pl\"ucker embedding via the Grassmanian $\op{Gr}(k+1,n+1)$
to $\mathbb{P}^{\binom{n+1}{k+1}}$.

Let us turn this into an explicit formula for the directrix in dimension $n+1=5$, i.e.\ for $S_{1,2}$. 
Here the directrix $\mathbb{P}^1$ corresponds to a plane $\langle p,q\rangle$ in the affine space $\R^5$,
and it has Pl\"ucker coordinates $\pi_{ij}=p_iq_j-p_jq_i$. Let us form the $4\times5$ matrix $A$ with rows
$p$, $q$, $r=\bigl(r_i/(\lambda-a_i)\bigr)$ and $s=\bigl(s_i/(\lambda-a_i)\bigr)$.
The condition on $p,q$ is $\op{rank}(A)<4$. To write it down, denote
 $$
\varpi(i,j,k,l)= a_ia_j(r_is_j-r_js_i)(r_ks_l-r_ls_k) +a_ka_l(r_ks_l-r_ls_k)(r_is_j-r_js_i),
 $$
and also for a cyclic permutation $(abcde)$ of $(12345)$ 
 $$
\Omega_a= \mathfrak{S}_{bcd}\bigl(\varpi(b,c,d,e)\bigr),
 $$
where $\mathfrak{S}$ denotes the cyclic summation over indicated indices.
Next, for the set of five indices define
 $$
\Psi_{ij} = \displaystyle(a_i - a_j)\prod_{k\neq i}(a_i-a_k)\prod_{k\neq j}(a_j-a_k).
 $$
Then the directrix is given by the relations for all $i<j$, $k<l$ in the projective space of dimension
$\binom{n+1}2=10$:
 $$
\Omega_i\Omega_j\Psi_{ij}\pi_{ij}=\Omega_k\Omega_l\Psi_{k,l}\pi_{k,l}.
 $$
In other words, $[\dots:\pi_{ij}:\dots]=[\dots:\Omega_i^{-1}\Omega_j^{-1}\Psi_{ij}^{-1}:\dots]\in
\mathbb{P}^{10}$.

\section{Symmetries of scroll structures}
\label{sec:s}

In this section we discuss projective automorphisms of scrolls $S_{k,l}$ (Section \ref{sec:scrollsym}) and point symmetry algebras
of the PDE systems governing involutive scroll structures (Section \ref{sec:sym}).

\subsection{Projective automorphisms of scrolls}
\label{sec:scrollsym}

Rational normal scrolls $S_{k,l}$ in $\mathbb{P}^n$ are ruled surfaces over $\mathbb{P}^1$ given by
an equation of the form $\alpha \,\omega(\lambda)+\beta \,\phi(\lambda)$, where $[\alpha:\beta]\in\mathbb{P}^1$
is a parameter on the ruling, $\omega(\lambda)$ and $\phi(\lambda)$ parametrize rational normal curves in
$\mathbb{P}^k$ and $\mathbb{P}^l$, respectively, $k+l=n-1=\op{deg}(S_{k,l})$, and $\lambda=[\lambda_0:\lambda_1]$
is a parameter on the base $\mathbb{P}^1$. In this section we will work over $\C$.

We recall that del Pezzo's theorem implies that an irreducible ruled surface  of degree $n-1$ spanning $\mathbb{P}^n$
is a scroll $S_{k,l}$, including $k=0$, i.e.\ a cone or the plane $\mathbb{P}^2$. (The only non-ruled exception is the 
Veronese surface, see \cite{Dol}.)

 \begin{prop}
The symmety algebra of $S_{k,l}$ is $\mathfrak{sl}_2\oplus\mathfrak{sl}_2$ for $k=l$ and
$\mathfrak{gl}_2\ltimes\C^{l-k+1}$ for $k<l$ where the semi-direct factor is the irreducible representation $S^{k-l}\C^2$. 
In particular, dimension of the symmety algebra is $l-k+5+\delta_{k,l}$.
 \end{prop}

 \begin{proof}
In the splitting $\C^{n+1}=\C^{k+1}\times\C^{l+1}$, the scroll $S_{k,l}$  is given as 
 $$
S_{k,l}=\mathbb{P}\bigl\{(\alpha \lambda_0^k,\alpha \lambda_0^{k-1}\lambda_1,\dots,
\alpha \lambda_0\lambda_1^{k-1},\alpha \lambda_1^k,\ 
\beta \lambda_0^l,\beta \lambda_0^{l-1}\lambda_1,\dots,
\beta \lambda_0\lambda_1^{l-1},\beta \lambda_1^l)\bigr\}. 
 $$
This variety is defined by the equations
 \begin{gather*}
x^ix^j-x^{i+1}x^{j-1}=0\quad (0\leq i<j\leq k),\qquad
y^iy^j-y^{i+1}y^{j-1}=0\quad (0\leq i<j\leq l),\\
x^iy^j-x^{i+1}y^{j-1}=0\quad (0\leq i<k,\ 0<j\leq l).
 \end{gather*}
Taking the general projective vector field,
 $$
X=\sum a_i^jx^i\p_{x^j}+\sum b_i^jy^i\p_{x^j}+\sum c_i^jx^i\p_{y^j}+\sum d_i^jy^i\p_{y^j},
 $$
differentiating the equations and setting the result equal to zero modulo the equations we get the following 
symmetry generators:
 \begin{gather*}
\sum_{i=0}^k y^{i+r}\p_{x^i}\quad (0\leq r\leq l-k),\\
\sum_{i=0}^{k-1}(k-i)x^{i+1}\p_{x^i}+\sum_{j=0}^{l-1}(l-j)y^{j+1}\p_{y^i},\quad
\sum_{i=0}^{k-1}(i+1)x^i\p_{x^{i+1}}+\sum_{j=0}^{l-1}(j+1)y^j\p_{y^{j+1}},\\ 
\sum_{i=0}^k(2i-k)x^i\p_{x^i}+\sum_{j=0}^l(2j-l)y^j\p_{y^j},\quad
\sum_{i=0}^kx^i\p_{x^i},\quad \sum_{j=0}^ly^j\p_{y^j}.
 \end{gather*}
The first $r+1$ generators correspond to the nilradical, the next four to $\mathfrak{gl}_2$, and the sum of 
the last two is the center of $\mathfrak{gl}_{n+1}$, which does not act in the projectivization and hence should be omitted.
For $k=l$ one should omit the first $r+1$ generators and add the following two:
 $$
\sum_{i=0}^kx^i\p_{y^i},\quad \sum_{i=0}^ky^i\p_{x^i}.
 $$
The structure relations are straightforward.
 \end{proof}

The automorphism groups mentioned in the introduction are obtained by exponentiation of the indicated Lie algebras.

 \begin{cor}
For $S_{k,k}$ and $S_{k-1,k}$, both symmetry algebras have dimension 6. 
In the first case the Lie algebra is semi-simple, while in the second it has a 3-dimensional radical.
 \end{cor}
 
Let us note that, abstractly, scrolls coincide with the Hirzebruch surfaces \cite{Reid}, which are $\mathbb{P}^1$
bundles over $\mathbb{P}^1$ given by
 $$
\mathbb{F}_n=\mathbb{P}(\mathcal{O}(0)\oplus\mathcal{O}(n)).
 $$
These surfaces are embedded into $\mathbb{P}^5$ via the Veronese embedding from
 $$
\mathbb{F}_n=\{([x:y:z],[u:v])\,|\,uv^n=zu^n\}\subset\mathbb{P}^2\times\mathbb{P}^1,
 $$
and the ruling is given by the projection $\pi([x:y:z],[u:v])=[u:v]$ with the fiber $F=\pi^{-1}(\cdot)$. 
The ``exceptional divisor'' $E_n=[1:0:0]\times\mathbb{P}^1$ corresponds to the directrix of the scroll.

Since $\op{Pic}(\mathbb{F}_n)=\langle[F],[E_n]:F^2=0,F\cdot E_n=1,E_n^2=-n\rangle$, one can compute the automorphism
group, cf.\ \cite{Blanc}, to be $\op{Aut}(\mathbb{F}_n)=GL_2\ltimes S^n\C^n$ for $n>0$ (the case $n=0$ corresponds to
the quadric). It is generated by the M\"obius transformations $[u:v]\mapsto[\alpha u+\beta v:\gamma u+\delta v]$ on
the base together with the following transformation of the fibre:
 $$
[x:y:z]=[xu^n:yu^n:yv^n]\mapsto[xu^n+y(a_0u^n+a_1u^{n-1}v+\dots+a_nv^n):y(\alpha u+\beta v)^n:y(\gamma u+\delta v)^n] 
 $$
in the chart $u\neq0$, and similarly in the other chart $v\neq0$.
It is instructive to note that the automorphisms of $\mathbb{F}_n\simeq S_{k,k+n}$ coincide, even though
this is an isomorphism of abstract surfaces and the second surface is treated in its standard projective embedding.

In this paper we focused on the following two surfaces,
 $$
\mathbb{F}_0=\mathbb{P}^1\times\mathbb{P}^1\simeq S_{k,k},\qquad
\mathbb{F}_1=\mathbb{P}^2\#\bar{\mathbb{P}}^2\simeq S_{k-1,k},
 $$
which were shown to be the characteristic varieties of integrable hierarchies of various heavenly-type equations.
Our results hold over both $\Bbbk=\C$ and $\Bbbk=\R$.

Note that over $\R$, in smooth category, $\mathbb{F}_m\simeq\mathbb{F}_n$ iff $m\equiv n\!\!\!\mod2$.
Thus, up to diffeomorphism, ruled surfaces are only 
the quadric $\mathbb{P}^1\times\mathbb{P}^1$ and the
blow-up of the projective plane $\mathbb{P}^2\#\bar{\mathbb{P}}^2$.
However, we consider real projective varieties, so that the type of the scroll $S_{k,l}$ is determined uniquely.

\subsection{Symmetry algebras of PDE systems governing involutive  scroll structures}
\label{sec:sym}

Symmetries can, in principle, reduce the count of functional freedom of ISS we derive,
so we have to explore symmetries of the corresponding integrable systems.

 \begin{prop}
Symmetry algebras for integrable equations governing involutive scroll structures depend
on a finite number of functions of two arguments (plus some more functions of one argument, as well as some constants).
 \end{prop}

The proof will be split into several lemmas. By Lie-B\"acklund theorem, most general symmetries of
systems of PDEs with more than one dependent variables are lifted from 0-jets.
We look for point symmetries in the general form
 $$
V=\sum_{i=1}^d a^i\p_{x^i} + \sum_{j=1}^m b^j\p_{u^j},
 $$ 
where $(x^i)_{i=1}^d$ are independent variables and $(u^j)_{j=1}^m$ are dependent variables; 
the coefficients $a^i,b^j$ initially depend on all of those. 

Below we specify this for $d=n+1$ being 5 and 6, and $m=3$ with dependent variables being $(u,v,w)$.
We substitute prolongation of this vector field to 2-jets, then Lie differentiate the defining equations of the system
and restrict the result to the system. Then we split the resulting expressions by 1- and remaining 2-jets,
and solve the obtained overdetermined system in each case. The results are as follows:

 \begin{lemma}
Functional freedom for symmetries of equations \eqref{triple} governing involutive $S_{1,2}$ scroll structures  (Construction 1 
from Section \ref{s12c1}) is 9 functions of 2 arguments.
 \end{lemma}

Explicitly (omitting functions of less than two arguments) we get 
 \begin{equation*}
 \begin{array}{rl}
a^1 =&\!\!\! 2x^1f^3 + x^2(x^3f^3_4 + f^4_4) + (x^3)^2f^3_5 + x^3(f^8+ f^4_5) + f^7,\\[3pt]
a^2 =&\!\!\! x^3f^1_5 + x^2(f^1_4 + f^3) + f^5,\ a^3 = x^3f^3 + f^4, \ a^4 = f^1,\ a^5=0,\\[3pt]
b^1 =&\!\!\! x^1f^1_5 + 2uf^3 +u f^1_4 + \tfrac12(x^3)^2f^1_{55} + x^3(x^2f^3_5 + x^2f^1_{45} + f^5_5) \\[3pt]
   +&\!\!\! \tfrac12(x^2)^2(f^1_{44} + 2f^3_4) + x^2(f^5_4 + f^8) + f^6,\\[3pt]
b^2 =&\!\!\! 3vf^3 + u\,(x^3f^3_4 + f^4_4) + x^1(3x^3f^3_5 + 2x^2f^3_4 + 2f^8 + f^4_5) 
	 +\tfrac12(x^2)^2(x^3f^3_{44} + f^4_{44}) \\[3pt]
   +&\!\!\! x^2((x^3)^2f^3_{45} + x^3f^4_{45} + x^3f^8_4 + f^7_4) 
	+ \tfrac12(x^3)^3f^3_{55} + (x^3)^2(f^8_5 + \tfrac12f^4_{55}) + x^3(f^7_5 + f^9) + f^2,\\[3pt]
b^3 =&\!\!\! u\,f^3_4 + 2wf^3 + x^1f^3_5 + \tfrac12(x^2)^2f^3_{44} + x^2(x^3f^3_{45} + f^8_4) + \tfrac12(x^3)^2f^3_{55}
	 + x^3f^8_5 + f^9,
 \end{array}
 \end{equation*}
where $f^i$ are functions of $(x^4,x^5)$ and $f^i_a,f^i_{bc}$ denote their derivatives by the indicated variables.

 \begin{lemma}
Functional freedom for symmetries of equations \eqref{alttriple} governing involutive $S_{1,2}$ scroll structures  (Construction 2 
from Section \ref{s12c2}) is 8 functions of 2 arguments.
 \end{lemma}

Explicitly (omitting functions of less than two arguments) we get 
 \begin{equation*}
 \begin{array}{rl}
a^1 =&\!\!\! f^1,\ a^2 = f^2,\ a^3 = f^3, \ a^4 = x^4f^4 + f^5,\ a^5=0,\\[3pt]
b^1 =&\!\!\! x^4f^3_5 + u\,(f^4 + f^3_3) + f^6,\\[3pt]
b^2 =&\!\!\! 2v\,f^4 + u\,(x^4f^4_3 + f^5_3) + (x^4)^2f^4_5 + x^4(f^5_5 - f^7) + f^8,\\[3pt]
b^3 =&\!\!\! -u\,f^4_3 + w\,f^4 - x^4f^4_5 + f^7,
 \end{array}
 \end{equation*}
where $f^1,f^2$ are functions of $(x^1,x^2)$ and $f^3,f^4,f^5,f^6,f^7,f^8$ are functions of $(x^3,x^5)$,
while $f^i_a,f^i_{bc}$ denote their derivatives by the indicated variables.

 \begin{lemma}
Functional freedom for symmetries of system (\ref{triple})+(\ref{triple1})
governing involutive $S_{2,2}$ scroll structures  (Construction 1 from Section \ref{s22c1}) is 11 functions of 2 arguments.
 \end{lemma}

Explicitly (omitting functions of less than two arguments) we get 
 \begin{equation*}
 \begin{array}{rl}
a^1 =&\!\!\! 
x^2f^2_6 + x^1(2f^3 + f^2_5 - 2f^1_6) + \tfrac12(x^4)^2f^2_{66} + x^4(x^3(f^3_6 + f^2_{56} - f^1_{66}) + f^5_6) \\[3pt]
	+&\!\!\!  \tfrac12(x^3)^2(2f^3_5 + f^2_{55} - 2f^1_{56}) + x^3f^6 + f^7,\\[3pt]
a^2 =&\!\!\! 
x^1f^1_5 + x^2(2f^3-f^1_6) + \tfrac12(x^3)^2f^1_{55} + x^3(x^4f^3_5 + f^4_5) + (x^4)^2(f^3_6 - \tfrac12f^1_{66}) \\[3pt]
	+&\!\!\! x^4(f^6 - f^5_5 + f^4_6) + f^9,\\[3pt]
a^3 =&\!\!\! 
x^4f^2_6 + x^3(f^3 + f^2_5 - f^1_6) + f^5, \
a^4 = x^3f^1_5 + x^4f^3 + f^4,\ 
a^5=f^2,\
a^6=f^1,
 \end{array}
 \end{equation*}
 \begin{equation*}
 \begin{array}{rl}
b^1 =&\!\!\! 
v\,f^1_5 + u\,(3f^3 -2f^1_6) + x^1(x^3f^1_{55} + x^4f^3_5 + f^4_5) + 
x^2(2x^3f^3_5 - x^3f^1_{56} + 3x^4f^3_6 - 2x^4f^1_{66} \\[3pt]
  +&\!\!\!  2f^6 -2f^5_5 + f^4_6) + \tfrac16(x^3)^3f^1_{555} + \tfrac12(x^3)^2(x^4f^3_{55} + f^4_{55}) 
+ x^3((x^4)^2(f^3_{56} - \tfrac12f^1_{566}) + f^9_5 \\[3pt]
   +&\!\!\!  x^4(f^4_{56} - f^5_{55} +f^6_5)) + (x^4)^3(\tfrac12f^3_{66} - \tfrac13f^1_{666}) + (x^4)^2(\tfrac12f^4_{66} 
- f^5_{56} + f^6_6) + x^4(f^9_6 + f^8) + f^{11},\\[3pt]
b^2 =&\!\!\! 
u\,f^2_6 + v\,(3f^3 + f^2_5 - 3f^1_6) + x^2(x^4f^2_{66} + x^3(f^3_6 + f^2_{56} - f^1_{66}) + f^5_6) \\[3pt] 
  +&\!\!\! x^1(x^3(f^2_{55} + 3f^3_5 -3f^1_{56}) + x^4(f^2_{56} + 2f^3_6 - 2f^1_{66}) + 2f^6 - f^5_5) 
+ \tfrac16(x^4)^3f^2_{666} \\[3pt]
  +&\!\!\! \tfrac12(x^4)^2(x^3(f^3_{66} + f^2_{566} - f^1_{666}) + f^5_{66}) 
+ x^4((x^3)^2(f^3_{56} - f^1_{566} + \tfrac12f^2_{556}) + x^3f^6_6 + f^7_6)  \\[3pt]
   +&\!\!\! \tfrac16(x^3)^3(3f^3_{55} -3f^1_{556} + f^2_{555}) + (x^3)^2(f^6_5 -\tfrac12f^5_{55}) + x^3(f^8 +f^7_5) + f^{10},\\ [3pt]
b^3 =&\!\!\! 
2w\,(f^3 - f^1_6) + x^1(f^3_5 - f^1_{56}) + x^2(f^3_6 - f^1_{66}) + \tfrac12(x^3)^2(f^3_{55} - f^1_{556}) \\[3pt]
 +&\!\!\!  x^3(x^4(f^3_{56} - f^1_{566}) - f^5_{55} + f^6_5)
  + \tfrac12(x^4)^2(f^3_{66} - f^1_{666}) + x^4(f^6_6 - f^5_{56}) + f^8,
 \end{array}
 \end{equation*}
where $f^i$ are functions of $(x^5,x^6)$ and $f^i_a,f^i_{bc}$ denote their derivatives by the indicated variables.

 \begin{lemma}
Functional freedom for symmetries of system (\ref{alttriple})+(\ref{triple2})
governing involutive $S_{2,2}$ scroll structures (Construction 2 from Section \ref{s22c2}) is 10 functions of 2 arguments.
 \end{lemma}

Explicitly (omitting functions of less than two arguments) we get 
 \begin{equation*}
 \begin{array}{rl}
a^1 =&\!\!\! f^2,\ a^2 = f^4,\ a^3 = x^4f^3_5 + x^3f^5 + f^6, \\[3pt]
a^4 =&\!\!\!  x^3f^1_6 + x^4(f^5 + f^1_5 - f^3_6) + f^7,\ a^5=f^1,\ a^6=f^3,\\[3pt]
b^1 =&\!\!\! 
v\,f^3_5 + u\,(2f^5 - f^3_6) + \tfrac12(x^4)^2f^3_{55} + x^4(x^3f^5_5 + f^6_5) + (x^3)^2(f^5_6 - \tfrac12f^3_{66}) 
+ x^3(f^6_6 - f^8) + f^9,\\[3pt]
b^2 =&\!\!\! 
u\,f^1_6 + v\,(2f^5 + f^1_5 - 2f^3_6) + \tfrac12(x^3)^2f^1_{66} + x^3(x^4(f^1_{56} + f^5_6 - f^3_{66}) + f^7_6) \\[3pt]
   +&\!\!\! \tfrac12(x^4)^2(f^1_{55} + 2f^5_5 - 2f^3_{56}) + x^4(f^7_5 - f^8) + f^{10},\\ [3pt]
b^3 =&\!\!\! 
w\,(f^5 - f^3_6) + x^3(f^3_{66} - f^5_6) + x^4(f^3_{56} - f^5_5) + f^8,
 \end{array}
 \end{equation*}
where $f^2,f^4$ are functions of $(x^1,x^2)$ and $f^1,f^3,f^5,f^6,f^7,f^8,f^9,f^{10}$ are functions of $(x^5,x^6)$,
while $f^i_a,f^i_{bc}$ denote their derivatives by the indicated variables.

\medskip

Similarly, one can show that symmetries of the second and general heavenly equations depend on  finitely many 
functions of at most two arguments. Thus, symmetry does not change the functional freedom for ISS 
established in this paper.

\section{Concluding remarks}

\begin{itemize}

\item Theorems \ref{t12} and \ref{t22} suggest that generic involutive $S_{k-1,k}$ and $S_{k,k}$ scroll structures are governed by higher flows of the hierarchy of conformal self-duality equations. Such structures should depend on $6k-3$ and $6k$ arbitrary functions of 3 variables, respectively. 
For $S_{1,1}$ (conformal structures in 4D) this count was first obtained in \cite{Grossman} and then re-confirmed in \cite{DFK} using a different approach. For $S_{1,2}$ and $S_{2,2}$ structures, this corresponds to Theorems \ref{t12}, \ref{t22} of the present paper. The general result can be obtained by similar, but more cumbersome, computations.

\item A hyper-K\"ahler hierarchy related to the second heavenly equation 
was discussed recently in \cite{BS}. The step of this hierarchy is not 1 but 4, as it describes hyper-K\"ahler metrics in dimension $4n$. 
The corresponding characteristic variety  has (affine) dimension $2n+1$ and degree $2n$, whence the general local solution depends on $2n+1$ arbitrary functions of $2n$ variables, giving the correct count for local functional freedom of hyper-K\"ahler metrics, cf. \cite{KP}.
In four dimensions, the hyper-K\"ahler hierarchy  reduces to the Plebansky second heavenly equation, which is due to the fact that  self-dual Ricci-flat structures in 4D are tantamount to hyper-K\"ahler metrics. 

\item Symmetry reduction of a self-dual metric in 4D by a Killing vector  is an Einstein-Weyl structure in 3D, as explored in \cite{Jones}. 
At the level of hierarchies, one obtains a reduction from the 4D hierarchy of conformal self-duality equations   to the 3D Einstein-Weyl hierarchy. In the language   of characteristic varieties, this corresponds to projections of scrolls along the rulings to curves in the spaces of one dimension less. 

\item Finally, our calculations suggest that scroll structures of type $S_{k,l}$ with $\vert k-l\vert \geq 2$ 
(in the simpest case, involutive $S_{1,3}$ structures), do not arise, as characteristic varieties, on solutions of the hierarchy 
of conformal self-duality equations or its reductions (which is surprising as all algebraic scrolls can be obtained by a projection 
from some $S_{k,k}$, ultimately, from $S_{\infty,\infty}$). Those ISS will be considered elsewhere.

\end{itemize}

\section*{Acknowledgements} We thank L. Bogdanov, A. Prendergast-Smith and A. Thompson for clarifying discussions.

\end{document}